\documentclass[letterpaper,11pt]{article}

\usepackage{amsmath, amsthm, amssymb, mathtools, bbm}

\usepackage{libertinus}

\newtheorem{theorem}{Theorem}[section]
\newtheorem{lemma}[theorem]{Lemma}

\newtheorem{corollary}[theorem]{Corollary}

\newtheorem{fact}{Fact}[section]

\newtheorem{definition}[theorem]{Definition}
\newtheorem{problem}[theorem]{Problem}
\newtheorem{remark}[theorem]{Remark}
\newtheorem{example}{Example}

\usepackage[margin=1in]{geometry}
\usepackage{fullpage}
\usepackage{setspace}
\allowdisplaybreaks

\usepackage[usenames, dvipsnames]{xcolor}
\definecolor{darkgreen}{rgb}{0.0, 0.39, 0.0}
\definecolor{crimsonglory}{rgb}{0,0,0} 

\usepackage{hyperref}
\hypersetup{
	colorlinks = true,
	urlcolor   = blue,
	linkcolor  = blue,
	citecolor  = darkgreen
}

\usepackage{graphicx}
\usepackage{float}
\usepackage{caption}
\usepackage{subcaption}
\usepackage{tikz}
\usetikzlibrary{shapes, arrows.meta, positioning, patterns, fit,
	arrows, snakes, automata, backgrounds, petri, calc,matrix}
\usepackage{pgfplots}
\usepackage{cleveref}
\usepackage{thm-restate}
\usepackage{booktabs} 
\usepackage{tabularx}
\usepackage{makecell} 
\usepackage{xcolor} 
\usepackage{ragged2e}
\pgfplotsset{compat=1.18}

\usepackage{enumerate}
\usepackage{enumitem,linegoal}
\renewcommand{\theenumi}{(\roman{enumi})}

\usepackage[noend]{algorithmic}
\usepackage[ruled,linesnumbered,noend,noline]{algorithm2e}

\SetCommentSty{mycommfont}
\SetKwInput{KwInput}{Input}
\SetKwInput{KwOutput}{Output}

\usepackage{threeparttable}

\usepackage[framemethod=TikZ]{mdframed}
\newmdenv[
backgroundcolor=gray!5,
linecolor=black,
linewidth=0.7pt,
roundcorner=6pt,
innertopmargin=12pt,
innerbottommargin=12pt,
innerleftmargin=10pt,
innerrightmargin=10pt,
skipabove=12pt,
skipbelow=12pt
]{fancybox}

\usepackage[draft]{changes} 
\definechangesauthor[name={Masoud}, color=teal]{ms}

\usepackage[normalem]{ulem} 

\newcommand{\BigONotation}{\mathcal{O}}
\newcommand{\LittleONotation}{o}

\newcommand{\pattern}{\textsf{P}}
\newcommand{\textstr}{\textsf{T}}
\newcommand{\lengthP}{\ensuremath{m}}
\newcommand{\lengthT}{\ensuremath{n}}
\newcommand{\wildcard}{\texttt{?}}
\newcommand{\alphabetSize}{\ensuremath{\sigma}}
\newcommand{\seth}{\textsf{SETH}}
\newcommand{\eps}{\ensuremath{\epsilon}}
\newcommand{\alphabet}{\ensuremath{\Sigma}}
\newcommand{\occ}{\Gamma} 

\newcommand{\concat}{\cdot}

\newcommand{\nonWild}{\omega}

\newcommand{\OV}{\text{Orthogonal Vectors }}

\newcommand{\OurProb}{Dynamic Pattern Matching with Wildcards}

\newcommand{\RangeProblem}{Range-Pair Problem}

\newcommand{\RangeAlphabet}{\Lambda}          
\newcommand{\RangeAlphabetSize}{\lambda}      
\newcommand{\RangeText}{\mathcal{T}}         
\newcommand{\RangeLen}{\eta}                    
\newcommand{\AppliedText}{\overline{\RangeText}} 

\newcommand{\BlockSize}{\beta}                   
\newcommand{\NumBlocks}{\left\lceil \frac{\RangeLen}{\BlockSize} \right\rceil} 

\newcommand{\BlockStart}[1]{\mathrm{beg}(#1)}
\newcommand{\BlockEnd}[1]{\mathrm{end}(#1)}

\newcommand{\LazyRebuild}{\mu}       

\newcommand{\PairMatrix}{\mathsf{Pair}}      
\newcommand{\NearMatrix}{\mathsf{Near}}      

\newcommand{\PairQuery}[5]{\texttt{PairQuery}(#1,#2,#3,#4,#5)} 
\newcommand{\UpdateSymbol}[2]{\texttt{UpdateSymbol}(#1,#2)}    

\newcommand{\MapSymbol}{\phi} 

\newcommand{\NumberOfMatches}{\texttt{cnt}}

\newcommand{\Pending}{\mathsf{Pending}}

\newcommand{\hashPrime}{p}
\newcommand{\hashBase}{b}
\newcommand{\hashFunc}[1]{H_{b,p}\left(#1\right)}
\newcommand{\hashPrimeFunc}[1]{H'_{b,p}\left(#1\right)} 
\newcommand*\samethanks[1][\value{footnote}]{\footnotemark[#1]}

\newcommand{\konote}[1]{}
\newcommand{\ignore}[1]{}

\newcounter{proccnt}

\makeatletter
\def\GrabProofArgument[#1]{ #1: \egroup\ignorespaces}
\def\proof{\noindent\textbf\bgroup Proof%
	\@ifnextchar[{\GrabProofArgument}{. \egroup\ignorespaces}}

\makeatother

\title{Dynamic  Pattern Matching with Wildcards}
\author{
	Arshia Ataee Naeini\thanks{School of Electrical and Computer Engineering, University of Tehran, Tehran, Iran} \and
	Amir-Parsa Mobed\samethanks \and
	Masoud Seddighin \thanks{Tehran Institute for Advanced Studies, Tehran, Iran}\and
	Saeed Seddighin
}
\date{} 

\begin{document}
	\maketitle
	\thispagestyle{empty}
	\sloppy
	
	\begin{abstract}
	We study the fully dynamic pattern matching problem where the pattern may contain up to $k$ wildcard symbols, each matching any symbol of the alphabet. Both the text and the pattern are subject to \emph{updates} (insert, delete, change). We design an algorithm with $\BigONotation(n \log^2 n)$ preprocessing and update/query time
$\tilde{\BigONotation}(kn^{k/{k+1}} + k^{2} \log n)$. The bound is truly sublinear for a constant $k$, and sublinear when		 $k = \LittleONotation(\log n)$.  
	We further complement our results with a conditional lower bound: assuming subquadratic preprocessing time, achieving truly sublinear update time for the case $k = \Omega(\log n)$ would contradict the Strong Exponential Time Hypothesis (\seth).
	Finally, we develop sublinear algorithms for two special cases:
	\begin{itemize}
		\item If the pattern contains $w$ non-wildcard symbols, we give an algorithm with preprocessing time $\BigONotation(nw)$ and update time $\BigONotation(w + \log n)$, which is truly sublinear whenever $w$ is truly sublinear.  
		\item Using FFT technique combined with block decomposition, we design a deterministic truly sublinear algorithm with preprocessing time $\BigONotation(n^{1.8})$ and update time $\BigONotation(n^{0.8} \log n)$ for the case that there are at most two non-wildcards.  
	\end{itemize}
\end{abstract}

\section{Introduction}

Pattern matching is a fundamental problem in computer science with applications in text search, bioinformatics, log analysis, and data mining \cite{Pandiselvam2014Comparative,Tahir2017ESWA,Faro2013CSUR,Abrishami2013Bioinformatics, IEEE8703383, IEEE6698319, IEEE8967097,Capacho2017,Cheng2013,Bushong2020RACS, Johannesmeyer2002AIChE}.
In its basic form, given a pattern \(\pattern\) of length \(m\) and a text \(\textstr\) of length \(n\) with symbols from alphabet $\alphabet$ with size $\sigma$, the task is to decide whether \(\pattern\) occurs in \(\textstr\) \emph{as a contiguous substring}; that is, whether there exists an index \(i \in \{1,\dots,n-m+1\}\) such that
$
\textstr_{i:i+m-1} = \pattern.
$
This problem has been studied for decades, and many efficient linear-time solutions such as the Knuth–Morris–Pratt and Boyer–Moore algorithms are well known \cite{knuth1977, boyer1977}. 

Traditionally, pattern matching aims to find exact occurrences of a pattern, yet in many practical applications such rigidity is unrealistic~\cite{Ukkonen1985,Landau1988,Navarro2001,Galil1990}. The pattern or the text may contain uncertain or corrupted {symbols} due to noise, ambiguity, or small variations; for example when accounting for mutations in genomic data. This motivates the study of algorithms that search for substrings close to the pattern under standard distance measures such as Hamming distance~\cite{abrahamson1987generalized,chan2023faster,jin2024shaving}.
Another common way to model uncertainty is to assume that the positions of the corrupted {symbols} are known. This is typically handled by introducing a \emph{wildcard} symbol, which matches any single {symbol} of the alphabet \(\Sigma\) \cite{Fischer1974,Indyk1998,Kalai2002}. This gives rise to the problem of \emph{pattern matching with wildcards}, where across both the pattern and the text, the total number of wildcard positions is at most \(k\).
Since each wildcard can stand for any symbol of \(\Sigma\), a pattern with \(k\) wildcards implicitly represents up to \(|\Sigma|^{k}\) different strings.

Adding wildcards makes the pattern matching problem  harder. With a fixed pattern, the search space is small. For example, let the text be $\textstr=\texttt{CACCGGCT}$ and the pattern be $\pattern=\texttt{CG}$. For this instance, there is only one match. If we instead use $\pattern'=\texttt{C?}$, where ``\texttt{?}'' refers to the wildcard symbol, then every occurrence of $\texttt{C}$ in $\textstr$ becomes a match. This simple change increases the number of matches from $1$ to $4$. Thus, even limited uncertainty can drastically increase the complexity of the problem.  

A substantial body of work has studied pattern matching with wildcards in either the pattern or the text~\cite{Kalai2002,bathie2024pattern,amir2004faster}.
For this problem, Fischer and Paterson~\cite{Fischer1974} gave an
$\BigONotation(n \log m \log \sigma)$ algorithm, and later Indyk~\cite{Indyk1998} gave a randomized $\BigONotation(n \log n)$ algorithm.
Kalai~\cite{Kalai2002} subsequently obtained a randomized $\BigONotation(n \log m)$ algorithm, while the first deterministic $\BigONotation(n \log m)$ algorithm was given by Cole and Hariharan~\cite{ColeHariharan2002}.


In this paper, we study pattern matching with wildcards in a fully dynamic setting.  
Both the text and the pattern may change through insertions, deletions, or substitutions, including updates to wildcard positions.  
Our goal is to maintain  pattern-matching queries in sublinear time.
For dynamic pattern matching, most prior work considers settings without wildcards. In these studies, 
dynamism may arise in several ways: the text may be updated dynamically, the pattern may be updated dynamically, or both may change over time. 
All three cases have been studied extensively. 
For dynamic text with a static pattern, Amir et al.~\cite{amir2007dynamic} presented an
algorithm that preprocesses a text of length $n$ and a pattern of length $m$ in $\BigONotation(n \log \log m + m \sqrt{\log m})$ time, reporting all new occurrences of the pattern after each text update in $\BigONotation(\log \log m)$ time.

For static text with a dynamic pattern, efficient suffix tree based solutions were introduced by Weiner~\cite{weiner1973linear} and later improved by McCreight~\cite{mccreight1976space}, Ukkonen~\cite{ukkonen1995line}, Farach~\cite{farach1997optimal}, and others~\cite{karkkainen2003simple}. 
For a fixed finite alphabet, these algorithms preprocess the text in $\BigONotation(n)$ time and answer pattern queries in $\BigONotation(m + \mathit{occ})$ time, where $\mathit{occ}$ is the number of occurrences of the pattern. 
Finally, in the fully dynamic setting where both text and pattern may change, several algorithms have been developed~\cite{GuFarachBeigel1994,ferragina1995fast,sahinalp1996efficient,alstrup2000pattern}. 
For example, Sahinalp and Vishkin~\cite{sahinalp1996efficient} support insertions and deletions in $\BigONotation(\log^3 n + m)$ time per update, with query times comparable to Weiner’s static algorithm. 
Alstrup, Brodal, and Rauhe~\cite{alstrup2000pattern} handle insertions, deletions, and substring moves in $\BigONotation(\log^2 n \log \log n \log^* n)$ time per update, and answer pattern searches in $\BigONotation(\log n \log \log n + m + \mathit{occ})$ time.

Our work extends these dynamic pattern matching studies to consider wildcards. We provide a randomized algorithm that, after $\BigONotation(n \log^2 n)$ preprocessing, supports updates and queries in amortized time 
$
\tilde{\BigONotation}\left(kn^{\frac{k}{k+1}} + k^{2} \log n\right),
$
which is sublinear for $k = \LittleONotation(\log n)$ and truly sublinear for constant $k$. We also show that this bound is almost tight: under SETH, no algorithm with subquadratic preprocessing time can achieve truly sublinear update time when $k = \Omega(\log n)$. 
In addition to the general result, we design specialized algorithms for two restricted settings, with the complexity expressed in terms of the number of non-wildcard symbols in the pattern. 

\begin{enumerate}
	\item When the wildcard positions are fixed, we obtain preprocessing in $\BigONotation(nw)$ time and support updates in time $\BigONotation(w + \log n)$ time, where \(w\) is the number of non-wildcard symbols. 
	
	\item When the pattern contains at most two non-wildcard symbols, we give a deterministic algorithm based on FFT and block decomposition techniques, with preprocessing time $\BigONotation(n^{1.8})$ and update time $\BigONotation(n^{0.8} \log n)$.
\end{enumerate}

\paragraph*{Organization of the paper.}
Section~\ref{sec:resultsTechs} gives a high-level overview of our techniques and main results. 
In Section~\ref{sec:preliminaries}, we introduce formal definitions and preliminaries. 
The general dynamic pattern matching with wildcards setting is studied in Section~\ref{sec:dynamic}, followed by algorithms for sparse pattern matching regimes in Section~\ref{sec:sparse}. 
Finally, in Section~\ref{sec:hardness} we establish a conditional lower bound on the update time of any algorithm for the dynamic pattern matching with wildcards problem.

\section{Results and Techniques}  
\label{sec:resultsTechs}

We study the problem of dynamic pattern matching in the presence of wildcard symbols.  
Our contributions fall into two main lines.

\noindent\textbf{General setting.}  
In the most flexible version of the problem, both the text and the pattern may contain wildcards, and updates may insert, delete, or substitute symbols (including wildcards).  
For this setting, we design randomized algorithms with preprocessing time $\BigONotation(n \log^2 n)$ and update/query time
$
\tilde{\BigONotation}\big(kn^{k/(k+1)} + k^{2} \log n\big).
$
The bound is truly sublinear for constant $k$, and remains sublinear whenever $k = o(\log n)$.  
This gives the first general framework for handling fully dynamic wildcards on both sides of the matching problem.

\noindent\textbf{Sparse patterns.}  
We also investigate regimes where the number of non-wildcards in the pattern is small.  
This additional structure enables more efficient algorithms:  
\begin{itemize}
	\item If the pattern contains at most two non-wildcard symbols, we obtain a deterministic algorithm based on FFT and block decomposition, with preprocessing time $\BigONotation(n^{1.8})$ and update time $\BigONotation(n^{0.8}\log n)$. Though the setting might seem simple, this is the most technical part of the paper, with the most challenging component captured by the \emph{Range-Pair} problem (Problem~\ref{prob:range-pair}), which we solve and use in this regime.
	\item If all wildcards are confined to the pattern and their positions are fixed in advance, we design a randomized algorithm with preprocessing time $\BigONotation(nw)$ and update time $\BigONotation(w+\log n)$, where $w$ is the number of non-wildcard positions.
\end{itemize}

Finally, we establish a conditional lower bound on the update and query time for any algorithm solving dynamic pattern matching with wildcards.

Together, these results give a unified picture of the problem’s complexity.  
Figure~\ref{fig:wildcard_spectrum} highlights the regimes where our upper and lower bounds apply. Also, for completeness and comparison, Table~\ref{tab:results_fixed} summarizes the state-of-the-art results in pattern matching alongside our contributions.  
As shown in Table~\ref{tab:results_fixed}, when the pattern contains no wildcards, dynamic exact matching can be solved with $\BigONotation(n \log^2 n)$ preprocessing and polylogarithmic update time~\cite{charalampopoulos2020dynamic}.
However, as noted, even a single wildcard can increase the number of matches.

\definecolor{lightblue}{RGB}{225, 235, 255}
\definecolor{lightgreen}{RGB}{225, 255, 235}
\definecolor{lightyellow}{RGB}{255, 255, 200}

%
%
%
%
%
%
%
%
%
%
%
%

\begin{table}[t!]
	\caption{Summary of Pattern Matching Results. For static problems, preprocessing refers to pattern preprocessing and query time refers to text scanning.}
	\centering
	\small
	\label{tab:results_fixed}
	\renewcommand{\arraystretch}{1.5}
	
	\scalebox{0.9}{
		\begin{tabularx}{\textwidth}{|>{\centering\arraybackslash}m{3.8cm}|>{\centering\arraybackslash}X|>{\centering\arraybackslash}X|}
			\hline
			\textbf{Problem Variant} & \textbf{Preprocessing Time} & \textbf{Algorithm/Query Time} \\
			\hline \hline
			
			\multicolumn{3}{|c|}{\textbf{Static Matching}} \\
			\hline
			
			Exact Matching \cite{knuth1977, boyer1977} & $\BigONotation(m)$ & $\BigONotation(n)$ \\
			\hline
			
			Wildcard Matching \cite{Fischer1974, Indyk1998, Kalai2002} & $\BigONotation(m \log m)$ & $\BigONotation(n \log m)$ \\
			\hline
			
			$k$-mismatches(constant |\alphabet|) \cite{abrahamson1987generalized, Clifford2007} & $\BigONotation(m \log m)$ & $\BigONotation(n \log m)$ \\
			\hline \hline
			
			\multicolumn{3}{|c|}{\textbf{Dynamic Exact Matching}} \\
			\hline
			
			Dynamic Text \cite{amir2007dynamic} & $\BigONotation(n \log\log m)$ & $\BigONotation(\log\log m)$ per update \\
			\hline
			
			Fully Dynamic \cite{alstrup2000pattern} & $\BigONotation(n\log^2{n})$ & $\BigONotation(\text{polylog } n)$ per query \\
			\hline \hline
			
			\multicolumn{3}{|c|}{\textbf{Dynamic Matching with Wildcards (This Work)}} \\
			\hline
			
			General (Randomized) \textbf{[Theorem~\ref{thm:dynamic_wildcard_general}]} & $\BigONotation(n \log^2 n)$ & $\tilde{\BigONotation}(kn^{\frac{k}{k+1}} + k^2\log n)$ \\
			\hline
			
			Sparse Pattern (at most 2 non-wildcards) \textbf{[Theorem~\ref{thm:at-most-two-nonwildcard}]} & $\BigONotation(n^{9/5})$ & $\BigONotation(n^{4/5}\log n)$ \\
			\hline
			
			Sparse Pattern (fixed wildcards) \textbf{[Theorem~\ref{thm:sparse-pattern}]} & $\BigONotation(w n)$ & $\BigONotation(w + \log n)$ \\
			\hline
			
			Conditional Lower Bound \textbf{[Theorem~\ref{thm:lower-bound}]} & --- & No $\BigONotation(n^{1-\varepsilon})$ for $k = \Omega(\log n)$ \\
			\hline
		\end{tabularx}
	} 
\end{table}
Our approach builds on a simple but very useful observation.  
If the pattern \(P\) contains a symbol that appears only a few times in the text \(T\), we can directly check those occurrences in the text and verify the aligned substrings (see Figure \ref{cjg}).  

\begin{figure}[t]
	\centering
	\scalebox{0.9}{
	\begin{tikzpicture}[
		>=stealth, 
		every node/.style={font=\small}
		]
		
		\def\spectrumlength{12}
		\def\logpoint{6} 
		\def\sublinpoint{5}
		\def\tsublinpoint{2.5}
		
		\shade[left color=green!40!white, right color=orange!70, right color=red!70] (0, 0) rectangle (\spectrumlength, 1);
		
		\draw[thick] (0,0) rectangle (\spectrumlength, 1);
		
		\draw[thick, dashed, black] (\tsublinpoint, 0) -- (\tsublinpoint, 1);
		\node[above, align=center, fill=white, fill opacity=0.8, text opacity=1, rounded corners=2pt] at (\tsublinpoint, 1.15) {Truly Sublinear \\ {\footnotesize$k = \BigONotation(1)$}};

		\draw[thick, dashed, black] (\sublinpoint, 0) -- (\sublinpoint, 1);
		\node[above, align=center, fill=white, fill opacity=0.8, text opacity=1, rounded corners=2pt] at (\sublinpoint, 1.15) {Sublinear \\ {\footnotesize$k = \LittleONotation(\log n)$}};
		
		\draw[thick, dashed, black] (\logpoint, 0) -- (\logpoint, 1);
		\node[above, align=center, fill=white, fill opacity=0.8, text opacity=1, rounded corners=2pt] at (\logpoint, -2.2) {No Truly Sublinear\\ (SETH-hard) \\ {\footnotesize$k = \Omega(\log n)$}};

		\draw [decorate, decoration={brace, amplitude=5pt, mirror}, yshift=-6pt] (0,0) -- (\sublinpoint,0);
		\node[below, align=center] at ({\sublinpoint/2}, -0.8) {General Algorithm \\ {\footnotesize$\tilde{\BigONotation}(kn^{k/(k+1)} + k^{2} \log n)$}};
		
		\def\fixedpoint{11.2}
		\draw[->, thick, black] (\fixedpoint, 1.2) -- (\fixedpoint, 1.7);
		\node[above, align=center] at (\fixedpoint, 1.9) {{\footnotesize Special Case 1} \\ {\footnotesize Fixed non-wildcards} \\ {\footnotesize$\BigONotation(w + \log n)$}};
		\draw[very thick, dashed, green] (\fixedpoint, 0) -- (\fixedpoint, 1);
		
		\def\twopoint{11.9}
		\draw[->, thick, black] (\twopoint, -0.2) -- (\twopoint, -0.7);
		\node[below, align=center] at (\twopoint, -0.9) {{\footnotesize Special Case 2} \\ {\footnotesize $\leq 2$ non-wildcards} \\ {\footnotesize$\BigONotation(n^{0.8} \log n)$}};
		\draw[very thick, dashed, green] (\twopoint, 0) -- (\twopoint, 1);
	\end{tikzpicture}
}
	\caption{The complexity landscape of fully dynamic pattern matching with 
		$k$ wildcards.}
	\label{fig:wildcard_spectrum}
\end{figure}
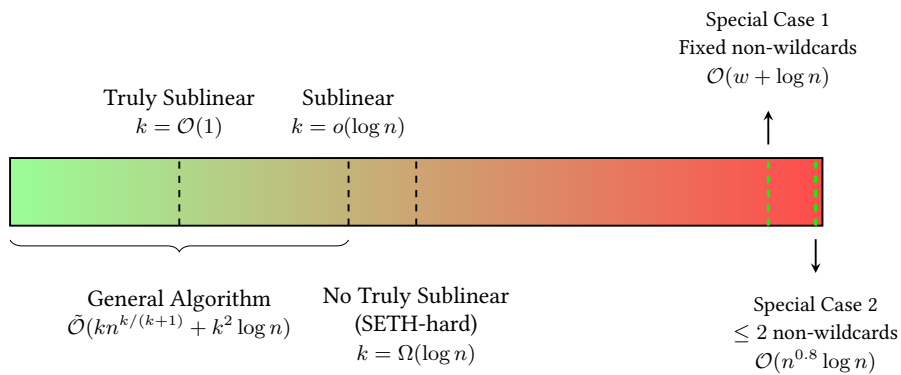

\begin{figure}[t]
	\centering
	\scalebox{0.65}{
		\begin{tikzpicture}[
			box/.style={draw, thick, rounded corners, minimum width=0.75cm, minimum height=0.75cm, font=\ttfamily},
			patternbox/.style={box, fill=blue!10, draw=blue!60},
			rarebox/.style={box, fill=yellow!20, draw=orange!70},
			textbox/.style={box, fill=green!10, draw=green!60},
			index/.style={font=\footnotesize, below, text=gray},
			matchrect/.style={draw=magenta!80, thick, dashed, rounded corners},
			arrow/.style={-stealth, thick, magenta!80},
			background/.style={fill=gray!5, draw=gray!20, rounded corners},
			>={Latex[length=2mm]}
			]
			
			\def\yPattern{3.6}
			\def\yText{1.6}
			\def\xstart{-0.6}
			
			\node[anchor=east] at (\xstart,\yPattern) {$P =$};
			\node[patternbox] (p1) at (1,\yPattern) {a}; \node[index] at (1,\yPattern-0.5) {1};
			\node[rarebox]    (p2) at (2,\yPattern) {b}; \node[index] at (2,\yPattern-0.5) {2};
			\node[patternbox] (p3) at (3,\yPattern) {c}; \node[index] at (3,\yPattern-0.5) {3};
			
			\node[anchor=east] at (\xstart,\yText) {$T =$};
			\node[textbox] (t1)  at (1,\yText) {a}; \node[index] at (1,\yText-0.5) {1};
			\node[rarebox] (t2)  at (2,\yText) {b}; \node[index] at (2,\yText-0.5) {2};
			\node[textbox] (t3)  at (3,\yText) {c}; \node[index] at (3,\yText-0.5) {3};
			\node[textbox] (t4)  at (4,\yText) {a}; \node[index] at (4,\yText-0.5) {4};
			\node[textbox] (t5)  at (5,\yText) {a}; \node[index] at (5,\yText-0.5) {5};
			\node[rarebox] (t6)  at (6,\yText) {b}; \node[index] at (6,\yText-0.5) {6};
			\node[textbox] (t7)  at (7,\yText) {c}; \node[index] at (7,\yText-0.5) {7};
			\node[textbox] (t8)  at (8,\yText) {a}; \node[index] at (8,\yText-0.5) {8};
			\node[textbox] (t9)  at (9,\yText) {d}; \node[index] at (9,\yText-0.5) {9};
			\node[textbox] (t10) at (10,\yText) {c}; \node[index] at (10,\yText-0.5) {10};
			
			\draw[arrow] (p2.south) .. controls +(0,-0.6) and +(0,0.8) .. (t2.north);
			\draw[arrow] (p2.south) .. controls +(0.6,-0.7) and +(0.1,0.9) .. (t6.north);
			
			\draw[matchrect] ($(t1.south west)+(-0.10,-0.12)$) rectangle ($(t3.north east)+(0.10,0.12)$);
			\draw[matchrect, color=cyan!80!black] ($(t5.south west)+(-0.10,-0.12)$) rectangle ($(t7.north east)+(0.10,0.12)$);
			
			
			\node[anchor=west] at (\xstart,\yText-1.45) {$\tau = 2$ \quad (rare $\le \tau$, frequent $> \tau$)};
			
			\node[patternbox, minimum width=0.6cm, minimum height=0.6cm, anchor=west] (Lpat) at (5.8,\yText-1.55) {};
			\node[anchor=west] at ($(Lpat.east)+(0.15,0)$) {pattern symbol};
			\node[rarebox, minimum width=0.6cm, minimum height=0.6cm, anchor=west] (Lrare) at (5.8,\yText-2.25) {};
			\node[anchor=west] at ($(Lrare.east)+(0.15,0)$) {rare (occurs $\le \tau$ times)};
			\node[textbox, minimum width=0.6cm, minimum height=0.6cm, anchor=west] (Ltext) at (5.8,\yText-2.95) {};
			\node[anchor=west] at ($(Ltext.east)+(0.15,0)$) {text symbol};			
			\begin{scope}[on background layer]
				\node[background, inner sep=6mm, fit=(current bounding box)] {};
			\end{scope}
			
		\end{tikzpicture}
	}
	\caption{Rare-symbol strategy: if the pattern contains a symbol that occurs at most $\tau$ times in the text (here, \(b\)), enumerate those occurrences and verify only the aligned substrings of length \(|P|\).}
	\label{cjg}
	
\end{figure}
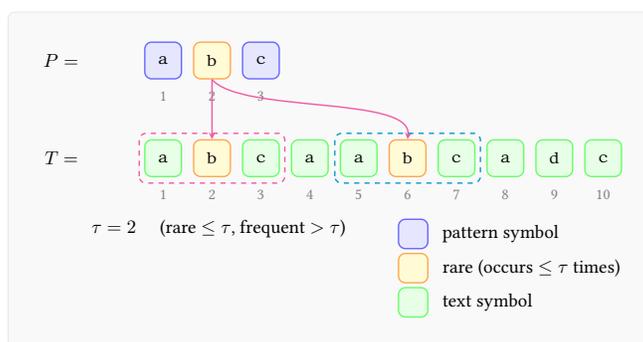

\noindent This leads us to a frequency-based classification of symbols: those that occur at most \(\tau\) times in \(T\) are called \emph{rare}, while the rest are considered \emph{frequent}.  
When \(P\) contains a rare symbol, we verify all possible candidate matches that align that rare symbol.  

Given this observation, we can assume without loss of generality that every symbol in the pattern is frequent. In this case, the number of candidate positions in the text that could match the pattern is large. The key insight, however, is that the number of distinct frequent symbols is limited to $n/\tau$. This restriction allows us to design an alternative algorithm based on standard data structures.

The main data structure we use in our algorithm is the \emph{polynomial rolling hash table}. At a high level, it provides a succinct summary of substrings that (i) updates efficiently after a single-symbol change and (ii) lets us handle wildcards. We treat wildcards by partitioning the string into contiguous non-wildcard segments: the hash is computed only over these segments, and two strings match under the wildcard semantics if and only if all corresponding non-wildcard parts have equal hashes. In this way, wildcards are naturally ignored, and correctness reduces to verifying equality of the fixed pieces. To support dynamic updates, we maintain these hashes in balanced \textsf{BSTs} such as treaps~\cite{seidel1996randomized}, which allow insertions, deletions, and substitutions in logarithmic time, and combine them with fully dynamic \textsf{LCS} to handle candidate completions with polylogarithmic overhead~\cite{charalampopoulos2020dynamic}.
%

As we discussed earlier, our algorithm achieves truly sublinear update time when the number of wildcards is constant, and remains sublinear as long as the number of wildcards is \(o(\log n)\).  
A natural question is whether these bounds can be improved further, perhaps matching the polylogarithmic update times known for the static case.  
Unfortunately, our results suggest otherwise.  
We show that once the number of wildcards reaches \(\Omega(\log n)\), no algorithm with reasonable preprocessing can achieve truly sublinear update time with subquadratic preprocessing unless the Strong Exponential Time Hypothesis (\textsf{SETH}) fails.  
This hardness result follows from a reduction from the orthogonal vectors problem.  

Given our results, several natural questions remain open:  
\begin{itemize}
	\item Can we design algorithms with truly sublinear update time when \(k\) is not constant but still \(o(\log n)\)? For constant \(k\), can we further improve the update and query bounds?  
	\item Can we obtain truly sublinear algorithms for cases where \(k = \Omega(\log n)\), provided the wildcards have some special structure?  
\end{itemize}

The first direction appears particularly promising for future work.  
In line with the second, we provide algorithms achieving truly sublinear update times for certain resticted settings.  
We first consider the setting where the wildcard positions are fixed in advance and the number of non-wildcard symbols is sublinear.  
In this case, it suffices to maintain rolling hash values only for the \(\BigONotation(k)\) non-wildcard positions, which enables sublinear update times for both the pattern and the text.  

For the special case where the pattern contains at most two non-wildcard symbols, the static version of the problem can already be solved efficiently by applying FFT: by encoding blocks of the text as polynomials and using convolution, one can align the two characters across the text. In the dynamic setting, however, a single update can change many convolution values, so the FFT method alone is insufficient. To address this, we combine block decomposition with the frequent/rare symbol technique: the text is divided into blocks, only the affected blocks are recomputed after an update, rare symbols are treated explicitly, and frequent ones are handled in the convolution structures. This approach ensures that each update can be processed within sublinear time, achieving preprocessing complexity \(\BigONotation(n^{\tfrac{9}{5}})\) and query complexity \(\BigONotation(n^{\tfrac{4}{5}})\).
	
\section{Related Work}

\label{sec:related}

String pattern matching has a long history.  
A key milestone was the \textbf{Knuth–Morris–Pratt (KMP)} algorithm~\cite{knuth1977}, which gave a linear-time solution via prefix-function preprocessing.  
Around the same time, the \textbf{Boyer–Moore} algorithm~\cite{boyer1977} introduced efficient skipping heuristics.  
These breakthroughs sparked a rich line of research
in string matching—including automaton-based methods such as Aho–Corasick ~\cite{aho1975efficient}, hashing strategies like Rabin–Karp~\cite{karp1987efficient}, and data-structure techniques based on suffix trees and arrays ~\cite{Weiner1973,wu1990np}.
Here we briefly mentionn some of these studies in several categories.

\smallskip
\noindent\textbf{Pattern matching with wildcards.}  
Wildcards add substantial flexibility but also complexity.  
Fischer and Paterson~\cite{Fischer1974} gave an early algorithm with \(\BigONotation(n \log m \log \sigma)\) runtime, later improved by Indyk and Kalai~\cite{Indyk1998,Kalai2002} and by Cole and Hariharan~\cite{cole2002} to \(\BigONotation(n \log m)\) deterministically.  
Clifford and Clifford~\cite{Clifford2007} presented a simpler algorithm with similar guarantees.  
More recent work includes filtering approaches by Barton et al.~\cite{Barton2014} and average-case analyses by Kopelowitz and Porat~\cite{Kopelowitz2022}, which separate the complexity of wildcard matching from exact matching.

\noindent\textbf{Dynamic pattern matching.}  
Weiner~\cite{Weiner1973} gave the first efficient solution for the case of a fixed text and a dynamic pattern, supporting queries in \(\BigONotation(|P| + \text{tocc})\) time after linear preprocessing.  
Amir et al.~\cite{Amir2007} studied the complementary case of a dynamic text with a fixed pattern, showing that new occurrences can be maintained in \(\BigONotation(\log \log m)\) time per update.  
More recently, Monteiro and dos Santos~\cite{Monteiro2024} revisited dynamic patterns with a suffix-array based framework, achieving amortized \(\BigONotation(\log |T|)\) update time for both symbol and substring edits.  

\smallskip
\noindent\textbf{Approximate matching.}  
Approximate matching has been studied extensively, with edit distance by Levenshtein~\cite{Levenshtein1966}, longest common subsequence (LCS) by Needleman and Wunsch~\cite{Needleman1970}, and dynamic time warping (DTW) analyzed by Bringmann and Künnemann~\cite{Bringmann2015} as the most common similarity measures.  
Ukkonen~\cite{Ukkonen1985} gave early algorithms for edit distance, and Masek and Paterson~\cite{Masek1980} introduced the Four-Russians speedup.  
However, conditional hardness results show strong barriers: Backurs and Indyk~\cite{Backurs2015} proved that strongly subquadratic edit distance would refute $\mathsf{SETH}$, and Bringmann and Künnemann~\cite{Bringmann2015} extended these lower bounds to LCS and DTW.  
A detailed survey of algorithmic techniques is given by Navarro~\cite{Navarro2001}.

\smallskip
\noindent\textbf{Dynamic and approximate wildcard matching.}  
Dynamic and approximate variants of wildcard matching have also been studied.  
Clifford et al.~\cite{Clifford2018} considered dynamic data structures for exact matching with wildcards, Hamming distance, and inner product (DynEM), establishing both upper and lower bounds.  
Crochemore et al.~\cite{Crochemore2015} gave algorithms for online matching with wildcards under updates, while Gog and Navarro~\cite{Gog2017} designed compressed-text versions supporting efficient updates.  
Charalampopoulos et al.~\cite{charalampopoulos2020dynamic} developed polylogarithmic-update algorithms for the dynamic longest common substring problem with wildcards, and Kociumaka, Radoszewski, and Starikovskaya~\cite{Kociumaka2019} proved SETH-hardness for the longest common substring with \(k\) mismatches while giving approximation algorithms.  
Most recently, Bathie, Charalampopoulos, and Starikovskaya~\cite{Bathie2024} provided a fine-grained analysis of approximate matching with wildcards and mismatches, presenting new structural lower bounds and efficient occurrence representations.

	\section{Preliminaries}
\label{sec:preliminaries}

We consider a finite alphabet \(\alphabet\) of size \(|\alphabet| = \alphabetSize\).  
Our input consists of two strings: a \emph{text} \(\textstr \in (\alphabet \cup \{\wildcard\})^*\) of length \(\lengthT\), and a \emph{pattern} \(\pattern \in (\alphabet \cup \{\wildcard\})^*\) of length \(\lengthP\).  
Both the text and the pattern together may contain at most \(k\) wildcard symbols, denoted by \(\wildcard\).
 A wildcard symbol can match any symbol from the alphabet.

For a string \(S\) and $1 \leq \ell \leq r \leq |S|$, we use \(S_{\ell:r}\) to denote the substring
$
S_{\ell} S_{\ell+1} \dots S_r.
$
For example, if \(S = \texttt{abcde}\), then \(S_{2:4} = \texttt{bcd}\).
We say that \(\pattern\) \emph{matches} a substring \(\textstr_{i:\, i+\lengthP-1}\) 
(\(1 \leq i \leq \lengthT - \lengthP + 1\)) if for every \(1 \leq j \leq \lengthP\), either 
\(\pattern_j = \textstr_{i+j-1}\), \(\pattern_j = \wildcard\) or $\textstr_{i+j-1}=\wildcard$.  
The set of all matching positions is denoted by
$$
\occ(\pattern, \textstr) =
\{\, i \mid 1 \leq i \leq \lengthT - \lengthP + 1,\ \pattern \text{ matches } \textstr_{i: i+\lengthP-1} \,\}.
$$

\smallskip
\noindent
The concatenation of two strings \(S\), \(S'\), denoted \(S \concat S'\), 
is defined by appending \(S'\) to the end of \(S\).  
For instance, if \(S = \texttt{ab}\) and \(S' = \texttt{cd}\), then 
\(S \concat S' = \texttt{abcd}\).

\smallskip
\noindent
We classify symbols in \(\alphabet\) as \emph{frequent} or \emph{rare} according to a frequency threshold.
Let the threshold value be
$$
\tau =
\bigl(n^k \log^{7} n \bigr)^{\tfrac{1}{k+1}}.
$$
A symbol \(c \in \alphabet\) is called \emph{rare} if it occurs fewer than \(\tau\) times in \(\textstr\), and \emph{frequent} otherwise.

\smallskip
\noindent
For each symbol \(c \in (\alphabet\ \cup \wildcard)\), we maintain the set
$
\mathfrak{R}(c) = \{\, i \mid 1 \leq i \leq \lengthT,\ \textstr_i = c \,\},
$
which stores all positions where \(c\) appears in \(\textstr\).  
In addition to membership queries for $\mathfrak{R}(c)$, we also require efficient support for \emph{lower bound queries}, i.e., finding the smallest index in \(\mathfrak{R}(c)\) that is greater than or equal to a given value. For handling the updates and queries efficiently, we use balanced binary search trees.

\begin{fact}[Balanced \textsf{BST} support]\label{fact:balanced-bst}
	Each set \(\mathfrak{R}(c)\) can be maintained under dynamic text updates such that membership queries, lower bound queries, and updates are all supported in \(\BigONotation(\log \lengthT)\) time, using a balanced binary search tree \cite{CLRS}.
\end{fact}

Many of our algorithms are randomized and succeed with high probability. Throughout this paper, we say an event occurs \textbf{with high probability} if it holds with probability at least $1 - 1/n$.

The foundation for our analysis, based on polynomial hashing, is detailed in \Cref{hashbased}. We use a polynomial rolling hash with a prime modulus \(\hashPrime\)~\cite{karp1987efficient} to efficiently support matching queries. The hash function is formally defined as follows.

\begin{definition}[Polynomial Rolling Hash]
	\label{def:hash_function}
	Let \(S \) be a string over an ordered alphabet \(\alphabet = \{c_1, c_2, \ldots, c_{|\alphabet|}\}\).  
	We define a mapping \(\mathrm{\pi} : \alphabet \to \mathbb{N}\) such that \(\mathrm{\pi}(c_j) = j\) for \(1 \leq j \leq |\alphabet|\).  
	The \emph{polynomial rolling hash} of \(S\) is
	$$
	\hashFunc{S} = \left( \sum_{i=1}^{|S|} \mathrm{\pi}(S_i) \cdot \hashBase^{\,|S|-i} \right) \bmod \hashPrime,
	$$
	where \(\hashBase\) is chosen uniformly at random from \(\{1,2,\dots,\hashPrime-1\}\), and \(\hashPrime\) is a large prime modulus.
\end{definition}

We leverage the following lemma from \Cref{hashbased} to prove our probabilistic claims in the paper. 

\begin{restatable}{lemma}{collisionprob}
	\label{lem:pairs-collision-union}
	Let \(\{(S_1,S'_1),\dots,(S_N,S'_N)\}\) be a collection of \(N\) pairs of distinct strings, 
	each string drawn from the alphabet \(\Sigma\) and of length at most \(x\).  
	Then the probability that there exists at least one pair \((S_i,S'_i)\) such that 
	\(\hashFunc{S_i}=\hashFunc{S'_i}\) is bounded by
	$
\tfrac{N\,x}{\hashPrime}.
	$
\end{restatable}

\paragraph*{Dynamic Updates} 
As described earlier, both the pattern and the text may change dynamically. We support the following operations:

\begin{itemize}
	\item \textbf{Pattern Update:} Insert, delete, or replace a symbol at position $i$ in the pattern, where the new symbol $c \in \alphabet \cup \{\wildcard\}$ for $1 \le i \le \lengthP$.
	\item \textbf{Text Update:} Insert, delete, or replace a symbol at position $i$ in the text, where the new symbol $c \in \alphabet \cup \{\wildcard\}$ for $1 \le i \le \lengthT$.
\end{itemize}

Updates are subject to the constraint that the total number of wildcards in $\pattern$ and $\textstr$ does not exceed $k$.  

After each update, the query asks whether there exists an index $i$ such that $\pattern$ matches the substring $\textstr_{i:i+\lengthP-1}$, i.e., whether $\occ(\pattern, \textstr) \neq \emptyset$.

\begin{example}
	Let the alphabet be \(\alphabet = \{a, b, c, d\}\) with \(\alphabetSize = 4\).  
	Initially, the pattern is \(\pattern = \texttt{a?b?c}\) of length \(\lengthP = 5\), and the text is \(\textstr = \texttt{aabbccba}\) of length \(\lengthT = 8\).  
	The pattern contains two wildcards (\(\pattern_2 = \pattern_4 = \wildcard\)).
	The symbol occurrence sets in the text are as follows:
	$\mathfrak{R}(a) = \{1, 2, 8\}, \mathfrak{R}(b) = \{3, 4, 7\}, \mathfrak{R}(c) = \{5, 6\}, \mathfrak{R}(d) = \emptyset.
	$
	By definition, a match occurs at position \(i\) if \(\pattern_1 = \textstr_i\), \(\pattern_3 = \textstr_{i+2}\), and \(\pattern_5 = \textstr_{i+4}\) (wildcards match any symbol).  
	Checking \(i = 1, \dots, 4\) gives
	$
	\occ(\pattern, \textstr) = \{1, 2\}.
	$
	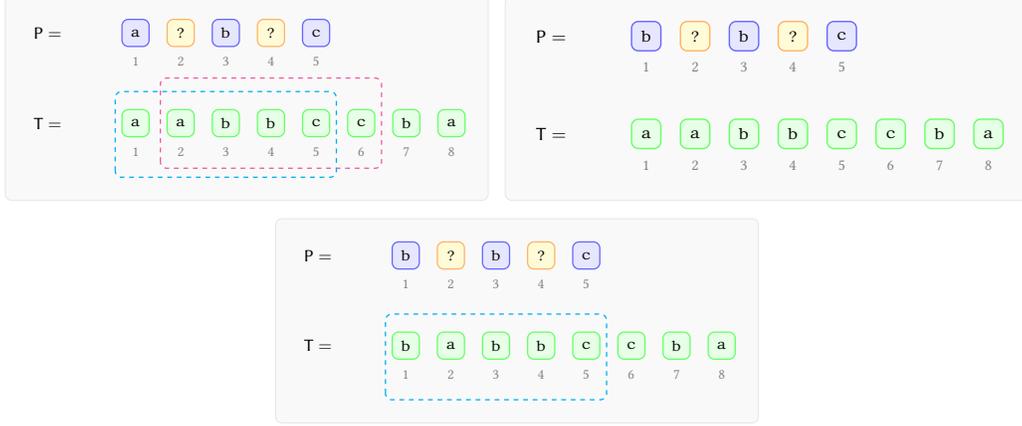
\begin{figure}[t]
		\centering
		\scalebox{0.6}{
			\begin{tikzpicture}[
				box/.style={draw, thick, rounded corners, minimum size=0.6cm, font=\ttfamily},
				patternbox/.style={box, fill=blue!10, draw=blue!60},
				wildcardbox/.style={box, fill=yellow!20, draw=orange!60},
				textbox/.style={box, fill=green!10, draw=green!60},
				index/.style={font=\footnotesize, below, text=gray},
				matchrect/.style={draw=magenta!80, thick, dashed, rounded corners},
				arrow/.style={-stealth, thick, magenta!80},
				background/.style={fill=gray!5, draw=gray!20, rounded corners}  
				]
				
				\node[anchor=east] at (-0.5, 2) {\(\pattern = \)};
				\foreach \i/\char in {1/a, 2/\texttt{?}, 3/b, 4/\texttt{?}, 5/c} {
					\begin{scope}[shift={(\i,2)}]
						\ifnum\i=2 \node[wildcardbox] (p\i) {\char};
						\else \ifnum\i=4 \node[wildcardbox] (p\i) {\char};
						\else \node[patternbox] (p\i) {\char};
						\fi\fi
						\node[index] at (0,-0.4) {\i};
					\end{scope}
				}
				
				\node[anchor=east] at (-0.5, 0) {\(\textstr = \)};
				\foreach \i/\char in {1/a, 2/a, 3/b, 4/b, 5/c, 6/c, 7/b, 8/a} {
					\begin{scope}[shift={(\i,0)}]
						\node[textbox] (t\i) {\char};
						\node[index] at (0,-0.4) {\i};
					\end{scope}
				}
				
				\draw[matchrect] (1.55,-1) rectangle (6.45,1);
				\draw[matchrect,color=cyan] (0.55,-1.2) rectangle (5.45,0.7);
				\begin{scope}[on background layer]
					\node[background, inner sep=5mm, fit=(current bounding box)] {};
				\end{scope}
				
			\end{tikzpicture}
		}
			\scalebox{0.65}{
			\begin{tikzpicture}[
				box/.style={draw, thick, rounded corners, minimum size=0.6cm, font=\ttfamily},
				patternbox/.style={box, fill=blue!10, draw=blue!60},
				wildcardbox/.style={box, fill=yellow!20, draw=orange!60},
				textbox/.style={box, fill=green!10, draw=green!60},
				index/.style={font=\footnotesize, below, text=gray},
				matchrect/.style={draw=magenta!80, thick, dashed, rounded corners},
				arrow/.style={-stealth, thick, magenta!80},
				background/.style={fill=gray!5, draw=gray!20, rounded corners}  
				]
				
				\node[anchor=east] at (-0.5, 2) {\(\pattern = \)};
				\foreach \i/\char in {1/b, 2/\texttt{?}, 3/b, 4/\texttt{?}, 5/c} {
					\begin{scope}[shift={(\i,2)}]
						\ifnum\i=2 \node[wildcardbox] (p\i) {\char};
						\else \ifnum\i=4 \node[wildcardbox] (p\i) {\char};
						\else \node[patternbox] (p\i) {\char};
						\fi\fi
						\node[index] at (0,-0.4) {\i};
					\end{scope}
				}
				
				\node[anchor=east] at (-0.5, 0) {\(\textstr = \)};
				\foreach \i/\char in {1/a, 2/a, 3/b, 4/b, 5/c, 6/c, 7/b, 8/a} {
					\begin{scope}[shift={(\i,0)}]
						\node[textbox] (t\i) {\char};
						\node[index] at (0,-0.4) {\i};
					\end{scope}
				}
				
				
				\begin{scope}[on background layer]
					\node[background, inner sep=5mm, fit=(current bounding box)] {};
				\end{scope}
				
			\end{tikzpicture}
		}
		
		\vspace{0.2cm}
		
				\scalebox{0.6}{
			\begin{tikzpicture}[
				box/.style={draw, thick, rounded corners, minimum size=0.6cm, font=\ttfamily},
				patternbox/.style={box, fill=blue!10, draw=blue!60},
				wildcardbox/.style={box, fill=yellow!20, draw=orange!60},
				textbox/.style={box, fill=green!10, draw=green!60},
				index/.style={font=\footnotesize, below, text=gray},
				matchrect/.style={draw=magenta!80, thick, dashed, rounded corners},
				arrow/.style={-stealth, thick, magenta!80},
				background/.style={fill=gray!5, draw=gray!20, rounded corners}  
				]
				
				\node[anchor=east] at (-0.5, 2) {\(\pattern = \)};
				\foreach \i/\char in {1/b, 2/\texttt{?}, 3/b, 4/\texttt{?}, 5/c} {
					\begin{scope}[shift={(\i,2)}]
						\ifnum\i=2 \node[wildcardbox] (p\i) {\char};
						\else \ifnum\i=4 \node[wildcardbox] (p\i) {\char};
						\else \node[patternbox] (p\i) {\char};
						\fi\fi
						\node[index] at (0,-0.4) {\i};
					\end{scope}
				}
				
				\node[anchor=east] at (-0.5, 0) {\(\textstr = \)};
				\foreach \i/\char in {1/b, 2/a, 3/b, 4/b, 5/c, 6/c, 7/b, 8/a} {
					\begin{scope}[shift={(\i,0)}]
						\node[textbox] (t\i) {\char};
						\node[index] at (0,-0.4) {\i};
					\end{scope}
				}
				
				\draw[matchrect,color=cyan] (0.55,-1.2) rectangle (5.45,0.7);
				\begin{scope}[on background layer]
					\node[background, inner sep=5mm, fit=(current bounding box)] {};
				\end{scope}
				
			\end{tikzpicture}
		}
		\caption{Pattern matching with \(\pattern = \texttt{a?b?c}\) and \(\textstr = \texttt{aabbccba}\). Yellow boxes represent wildcards. The red dashed region shows the match starting at position \(i=2\).}
		\label{fig:pattern_match}
	\end{figure}
	Now, we modify the first symbol of the pattern from \(\texttt{a}\) to \(\texttt{b}\).  
	This produces the new pattern \(\pattern = \texttt{b?b?c}\), which still contains exactly two wildcard symbols.  
	If we check the updated pattern against the text, we find that no substring of the text matches, so the set of occurrences becomes empty.
	\begin{figure}[t]
		\centering
	
		\label{fig:pattern_update}
	\end{figure}
	Next, we change the first symbol of the text from \(\texttt{a}\) to \(\texttt{b}\), resulting in \(\textstr = \texttt{babbccba}\).  
	This modification updates the occurrence sets as follows:
	$
	\mathfrak{R}(a) = \{2, 8\}, 
	\mathfrak{R}(b) = \{1, 3, 4, 7\}, 
	\mathfrak{R}(c) = \{5, 6\}, 
	\mathfrak{R}(d) = \emptyset.
	$
	Checking the updated text against the current pattern \(\texttt{b?b?c}\) shows that a match now exists at position \(i = 1\). See Figure \ref{fig:pattern_match}.
\end{example}

\subsection{Queries}
Throughout our algorithm, we frequently use two types of queries: the \emph{matching query} and the \emph{longest common substring (\textsf{LCS}) query}.  
Below, we define each query and state the time complexity required to answer it.

\vspace{0.2cm}
\noindent\textbf{Matching Query.}
The matching query is defined as
$
\texttt{is\_match}(\pattern_{\ell_1:r_1},\, \textstr_{\ell_2:r_2}),
$
which returns \texttt{true} if the pattern substring \(\pattern_{\ell_1:r_1}\) matches the text substring \(\textstr_{\ell_2:r_2}\), allowing wildcards in both \(\pattern\) and \(\textstr\).

\begin{restatable}{lemma}{hashmatch}
	\label{thm:is_match_hash_query}
	The query \(\texttt{is\_match}(\pattern_{l_1:r_1},\, \textstr_{l_2:r_2})\) can be answered with 
	\(\BigONotation(n)\) preprocessing time, \(\BigONotation(\log n)\) time per update operation, and 
	\(\BigONotation(k' \log n)\) time per query, where \(k'\) is the total number of wildcards 
	in \(\pattern_{l_1:r_1}\) and \(\textstr_{l_2:r_2}\).  
	The algorithm reports a match with probability \(1\) when one exists, and reports no match with probability 
	\(1 - \frac{1}{n^{2}}\) when no match exists.
\end{restatable}

%
%
%

\noindent\textbf{LCS Query.}  
The LCS query is defined as  
$
\texttt{LCS}(A, B),
$
which returns the length of the longest common substring of two dynamic strings \(A\) and \(B\).
By \cite{charalampopoulos2020dynamic}, we know that this query can be answered with \(\BigONotation(n \log^2 n)\) preprocessing time and \(\BigONotation(\log^8 n)\) time per update. 

\begin{theorem}[\cite{charalampopoulos2020dynamic}]
	\label{thm:lcs_query}
	For strings \(A\) and \(B\), each of length at most \(N\), the query 
	\(\texttt{LCS}(A, B)\) can be supported with \(\BigONotation(N \log^2 N)\) preprocessing time and 
	\(\BigONotation(\log^8 N)\) time per update operation. 
	Here an update operation may be a substitution, insertion, or deletion applied to either \(A\) or \(B\).
\end{theorem}


\section*{Handling Matching Query}\label{hashbased}


\noindent \textbf{{Collision Probability.}} A \emph{collision} occurs when two distinct strings \(S_1 \neq S_2\) yield the same hash value, i.e., \(\hashFunc{S_1} = \hashFunc{S_2}\). Since our hash function is chosen randomly from a universal family, we can bound the probability of such an event.

\begin{lemma}[Pairwise Collision Probability]
	\label{lem:pairwise-collision}
	For any two distinct strings \(S_1\) and \(S_2\), the probability that they collide under the hash function from Definition~\ref{def:hash_function} is
	$$
	\Pr[\hashFunc{S_1} = \hashFunc{S_2}] \;\leq\; \frac{\max(|S_1|,|S_2|)}{\hashPrime}.
	$$
	The probability is over the random choice of the base \(\hashBase\).
\end{lemma}

\begin{proof}
	The hash values are computed as:
	$$
	\hashFunc{S_1} = R_{S_1}(\hashBase), \quad 
	\hashFunc{S_2} = R_{S_2}(\hashBase) 
	$$
	where \(R_{S_1}(x)\) and \(R_{S_2}(x)\) are polynomials of degrees \(|S_1|-1\) and \(|S_2|-1\) respectively.
	
	Since \(S_1 \neq S_2\), the polynomial \(Q(x) = R_{S_1}(x) - R_{S_2}(x)\) is a non-zero polynomial of degree at most \(\max(|S_1|,|S_2|)-1\). 
	
	The hashes collide when \(Q(\hashBase) \equiv 0 \pmod{\hashPrime}\), i.e., when \(\hashBase\) is a root of \(Q(x)\) modulo \(\hashPrime\).
	
	By the Schwartz--Zippel lemma~\cite{schwartz1980fast}, a non-zero polynomial of degree \(d\) over a finite field has at most \(d\) roots. Since we work modulo prime \(\hashPrime\), we have:
	$$
	\Pr\!\left[\hashFunc{S_1} = \hashFunc{S_2}\right] 
	\;\le\; \frac{\max(|S_1|,|S_2|)-1}{\hashPrime} 
	\;<\; \frac{\max(|S_1|,|S_2|)}{\hashPrime}.
	$$
\end{proof}

\collisionprob*

\begin{proof}
	For each \(i\in[N]\), let \(C_i\) denote the event \(\hashFunc{S_i}=\hashFunc{S'_i}\).
	By Lemma~\ref{lem:pairwise-collision}, for each \(i\),
	\(\Pr[C_i]\le \tfrac{x}{\hashPrime}\) since both strings have length at most \(x\).
	By the union bound,
	$$
	\Pr\!\left[\bigcup_{i=1}^{N} C_i\right]
	\;\le\; \sum_{i=1}^{N} \Pr[C_i]
	\;\le\; \frac{N\,x}{\hashPrime}.
	$$
\end{proof}

\noindent\textbf{{Parameter choice.}}
In our application, for each alignment of the pattern \(P\) against the text we consider pairs 
\((S_i, S'_i)\) defined as follows. For the \(i\)-th alignment, let
$$
X_i = T[i : i+m-1]
$$
be the length-\(m\) substring of the text. Let the wildcard positions in \(P\) be 
\(w_1,\dots,w_{k'}\), and the wildcard positions in \(X_i\) be 
\(W^{(i)}_1,\dots,W^{(i)}_x\). Define 
$$
U_i = \{w_1,\dots,w_{k'}\}\cup\{W^{(i)}_1,\dots,W^{(i)}_x\},
$$
and sort \(U_i\) as \(u^{(i)}_1 < u^{(i)}_2 < \cdots < u^{(i)}_r\). 
We then split both \(P\) and \(X_i\) at these cut points and remove the wildcard symbols, concatenating 
the remaining fixed blocks.

Concretely,
$$
S_i = P[1:u^{(i)}_1-1] \concat P[u^{(i)}_1+1:u^{(i)}_2-1] \concat \cdots \concat P[u^{(i)}_r+1:m],
$$
$$
S'_i = X_i[1:u^{(i)}_1-1] \concat X_i[u^{(i)}_1+1:u^{(i)}_2-1] \concat \cdots \concat X_i[u^{(i)}_r+1:m].
$$

Thus \((S_i,S'_i)\) represents the pair obtained from the alignment of \(P\) with \(X_i\) after 
removing wildcards in the same positions.  
By Lemma~\ref{lem:pairs-collision-union}, across all \(n-m+1\) alignments the probability that 
any such pair collides is bounded by
$$
\frac{(n-m+1)\,m}{\hashPrime} \;\le\; \frac{n^2}{\hashPrime}.
$$
Choosing \(\hashPrime > n^3\) therefore guarantees that the collision probability in our 
problem is at most \(1/n\).

\noindent\textbf{{Computational complexity.}} Polynomial rolling hash functions have several computational properties that make them well-suited for our problem.  
By precomputing \(\hashBase^k \bmod \hashPrime\) for \(0 \leq k \leq |S|\) in \(\BigONotation(|S|)\) time, we can:
\begin{itemize}
	\item Compute all prefix hashes \(\hashFunc{S_{1:i}}\) in \(\BigONotation(|S|)\) time.
	\item Retrieve the hash of any substring \(S_{l:r}\) in \(\BigONotation(1)\) time using the rolling-hash property.
\end{itemize}
An additional useful feature is that the hash of a concatenated string \(\hashFunc{S \concat S'}\) can be derived directly from the hashes of \(S\) and \(S'\) in constant time, as we formalize in Lemma~\ref{lem:concat_hash_efficiency}.  

\begin{lemma}
	\label{lem:concat_hash_efficiency}
	Let \(S\) and \(S'\) be two strings over the same ordered alphabet, with lengths \(|S| = \ell_1\) and \(|S'| = \ell_2\).  
	Let \(\hashFunc{\cdot}\) be the polynomial rolling hash function from Definition~\ref{def:hash_function}.  
	If \(\hashBase^i \bmod \hashPrime\) are precomputed for all \(0 \leq i \leq \max(\ell_1, \ell_2)\), then \(\hashFunc{S \concat S'}\) can be computed from \(\hashFunc{S}\) and \(\hashFunc{S'}\) in \(\BigONotation(1)\) time.
\end{lemma}

\begin{proof}
	By the definition of the polynomial rolling hash, we have:
	$$
	\hashFunc{S} = \sum_{i=1}^{\ell_1} \pi(S_{i}) \cdot \hashBase^{\,\ell_1-i} \bmod \hashPrime,
	$$
	$$
	\hashFunc{S'} = \sum_{j=1}^{\ell_2} \pi(S'_{j}) \cdot \hashBase^{\,\ell_2-j} \bmod \hashPrime.
	$$
	The concatenation \(S \concat S'\) has length \(\ell_1 + \ell_2\), and its hash is:
	$$
	\hashFunc{S \concat S'}
	= \left( \sum_{i=1}^{\ell_1} \pi(S_{i}) \cdot \hashBase^{\,\ell_1+\ell_2-i}
	+ \sum_{j=1}^{\ell_2} \pi(S'_{j}) \cdot \hashBase^{\,\ell_2-j} \right) \bmod \hashPrime.
	$$
	The first term can be factored as:
	$$
	\left( \sum_{i=1}^{\ell_1} \pi(S_{i}) \cdot \hashBase^{\,\ell_1-i} \right) \cdot \hashBase^{\ell_2}
	= \hashFunc{S} \cdot \hashBase^{\ell_2} \bmod \hashPrime.
	$$
	Thus,
	$$
	\hashFunc{S \concat S'}
	= \left( \hashFunc{S} \cdot \hashBase^{\ell_2} + \hashFunc{S'} \right) \bmod \hashPrime.
	$$
	
	If \(\hashBase^{\ell_2} \bmod \hashPrime\) is precomputed, this formula requires only one modular multiplication, one modular addition, and one modular reduction, all of which take \(\BigONotation(1)\) time.  
	Therefore, the concatenation hash can be computed in \(\BigONotation(1)\) time.
\end{proof}

To extend the polynomial rolling hash to patterns containing wildcards,  
we define a modified hash function that ignores wildcard positions and hashes only the remaining symbols.  
This allows us to compare strings against wildcard patterns using efficient hash computations.

\begin{lemma}
	\label{lem:range_query_efficiency}
	Given a string $S$ of length $\ell$, we can preprocess $S$ in $\BigONotation(\ell)$ time so that, thereafter, two operations are supported throughout any sequence of changes: (i) \emph{range-hash queries}, which on input $1\le l\le r\le \ell$ return $\hashFunc{S_{l:r}}$ in $\BigONotation(\log \ell)$ time; and (ii) \emph{point updates}, which replace a single symbol $S_i\leftarrow \sigma$ in $\BigONotation(\log \ell)$ time. 
\end{lemma}

\begin{proof}
	We build a segment tree~\cite{CLRS} $T$ over the index set $\{1,\dots,\ell\}$. For each node $v$ that represents an interval $I_v=[l_v,r_v]$, maintain its length $\mathrm{len}(v)=r_v-l_v+1$ and the hash $\hashFunc{S_{l_v:r_v}}$. Before building the tree, precompute the powers $\hashBase^{k}\bmod \hashPrime$ for all $0\le k\le \ell$. This table is computed by a simple linear pass and uses $\BigONotation(\ell)$ time and space.
		The tree is constructed bottom–up. At the leaves we have $I_v=[i,i]$ and we set $\hashFunc{S_{i:i}}$. For an internal node $v$ with left child $u$ and right child $w$ (so $I_v=I_u\circ I_w$ and $\mathrm{len}(v)=\mathrm{len}(u)+\mathrm{len}(w)$), compute
	$$
	\hashFunc{S_{l_v:r_v}}\;=\;\big(\hashFunc{S_{l_u:r_u}}\cdot \hashBase^{\mathrm{len}(w)} + \hashFunc{S_{l_w:r_w}}\big)\bmod \hashPrime,
	$$
	which is exactly the concatenation rule from Lemma~\ref{lem:concat_hash_efficiency}. Since the segment tree has $\BigONotation(\ell)$ nodes and each node’s hash is obtained from its two children by a constant-time formula, the total preprocessing time is $\BigONotation(\ell)$.
	
	To handle a point update $S_i\leftarrow \sigma$, change the leaf hash at $[i,i]$ to $\hashFunc{\sigma}$ and then recompute the hash at every ancestor on the unique path to the root using the same concatenation formula. The height of the tree is $\BigONotation(\log \ell)$, and each recomputation is constant time, so the update takes $\BigONotation(\log \ell)$ time.
	
	To answer a range-hash query $\hashFunc{S_{l:r}}$, decompose $[l,r]$ into the usual $\BigONotation(\log \ell)$ canonical segments of the tree listed from left to right. Sweep once over these segments while maintaining an accumulator $\hashFunc{\cdot}_{\mathrm{acc}}$ initialized to $0$. When processing a segment represented by node $v$, update
	$$
	\hashFunc{\cdot}_{\mathrm{acc}} \;\gets\; \big(\hashFunc{\cdot}_{\mathrm{acc}}\cdot \hashBase^{\mathrm{len}(v)} + \hashFunc{S_{l_v:r_v}}\big)\bmod \hashPrime.
	$$
	By repeated application of the same concatenation rule, after the sweep $\hashFunc{\cdot}_{\mathrm{acc}}$ equals $\hashFunc{S_{l:r}}$. The sweep touches $\BigONotation(\log \ell)$ segments and performs $\BigONotation(1)$ work per segment, so the query runs in $\BigONotation(\log \ell)$ time.
\end{proof}

\hashmatch*
	
	\begin{proof}
		Suppose there are \( k' \) wildcard positions inside \([l_1, r_1]\), written in increasing order as \( w_1 < \cdots < w_{k'} \). These wildcards partition \([l_1, r_1]\) into at most \( k' + 1 \) blocks. Define for \( j = 0, \dots, k' \):
		$$
		s_j \;=\; \begin{cases}
			l_1 & j = 0,\\
			w_j + 1 & j \ge 1,
		\end{cases}
		\qquad
		e_j \;=\; \begin{cases}
			w_{j+1} - 1 & j < k',\\
			r_1 & j = k',
		\end{cases}
		$$
		and keep only those indices with \( s_j \le e_j \) (empty blocks arise if two wildcards are adjacent and can be skipped). In words, the blocks are
		$
		P_{s_0:e_0}, \; P_{s_1:e_1}, \; \dots, \; P_{s_{k'}:e_{k'}},
		$
		each containing only fixed symbols of the pattern. Under the alignment \( P_{l_1:r_1} \leftrightarrow T_{l_2,r_2} \), the \( j \)-th pattern block corresponds to the text block
		$$ T_{s'_j:e'_j} = T_{l_2 + (s_j - l_1):l_2 + (e_j - l_1)} $$
		because we match positions by equal offsets from the left endpoints.
		
		We test each block pair via hashing. By Lemma~\ref{lem:range_query_efficiency}, after \( \BigONotation(\ell) \) preprocessing we can obtain \(\hashFunc{P_{s_j:e_j}}\) and \(\hashFunc{T_{s'_j:e'_j}}\) in \( \BigONotation(\log \ell) \) time per block. Comparing the two hashes then certifies equality of that block with collision probability for the chosen modulus. There are at most \( k' + 1 \) non-empty blocks, so the total time is \( \BigONotation((k' + 1) \log \ell) = \BigONotation(k' \log \ell) \).
		
		We verify the overall alignment by checking whether $$
		\hashFunc{P_{s_0:e_0} \concat P_{s_1:e_1} \concat \dots \concat P_{s_{k'}:e_{k'}}} = \hashFunc{T_{s'_0:e'_0} \concat T_{s'_1:e'_1} \concat \dots \concat T_{s'_{k'}:e'_{k'}}}.
		$$ By Lemma~\ref{lem:pairwise-collision}, the probability of a collision for this hash is at most \( \frac{1}{n^{2}} \).
	\end{proof}

	\section{Fully Dynamic Pattern Matching with Wildcards}

\label{sec:dynamic}

In this section, we consider the general case of \emph{dynamic pattern matching with wildcards problem}. Note that here both the text \(\textstr\) and the pattern \(\pattern\) may contain wildcard symbols. We assume that updates consist only of single-character substitutions, and do not include insertions or deletions. 
This simplifies the presentation of the main ideas. 
We will later explain in Remark~\ref{rem:insdel} how the algorithm extends to handle insertions and deletions as well. 

Our solution is based on maintaining a modified version of the text, denoted by \textsf{T'}, where every rare symbol is replaced with a special placeholder `\#`. Therefore, \textsf{T'} is defined as
$$
\textsf{T'}_i \;=\;
\begin{cases}
	\# & \text{if } \textsf{T}_i \text{ is rare}, \\
	\textsf{T}_i & \text{if } \textsf{T}_i \text{ is frequent}.
\end{cases}
$$
Note that the frequency of symbols may change during the dynamic updates.
Whenever the frequency of a symbol crosses the threshold \(\tau\), we update \textsf{T'} by replacing its occurrences with `\#` (if it becomes rare) or restoring them to the original symbol (if it becomes frequent). To implement this, we maintain for each symbol \(c\in\Sigma\) the occurrence set \(\mathfrak{R}(c)\). When the text is updated at position \(i\) from \(c_{\text{old}}\) to \(c_{\text{new}}\), we remove \(i\) from \(\mathfrak{R}(c_{\text{old}})\) and insert it into \(\mathfrak{R}(c_{\text{new}})\). If the frequency of a frequent symbol falls below \(\tau\), we mark its \(\tau\) occurrences in \textsf{T'} with the placeholder `\#`; conversely, if a rare symbol becomes frequent, we restore its actual symbol at those positions. Each such change involves only \(\BigONotation(\tau)\) positions, giving an update cost of \(\BigONotation(\tau)\). The full procedure is formalized in Algorithm~\ref{alg:dynamic-maintenance}.

The preprocessing phase (Algorithm~\ref{alg:dynamic-preprocess}) initializes the data structures, including building \textsf{T'}, maintaining occurrence sets \(\mathfrak{R}(\cdot)\) for each symbol, and computing the frequent set \(\textsf{F}\). After each update, Algorithm~\ref{alg:dynamic-maintenance} efficiently maintains these structures.

Given this setup, the matching procedure proceeds differently depending on whether the pattern contains a rare symbol. If it does (Case~1, Algorithm~\ref{alg:dynamic-case1}), then any valid match must align that rare symbol either with one of its actual occurrences in the text or with one of the wildcard positions in the text. This observation allows us to restrict attention to candidate substrings centered around these positions in \textsf{T}, checking them directly for consistency with the rest of the pattern. 

If, on the other hand, the pattern contains no rare symbols (Case~2, Algorithm~\ref{alg:dynamic-case2}), then each wildcard may be filled with either a frequent symbol or the placeholder `\#`. We systematically enumerate all possible completions of both the pattern and the text under this rule and verify each candidate against \textsf{T'} using efficient \textsf{LCS} queries. Since the number of completions depends only on the number of wildcards and the frequency threshold, this approach avoids recomputation from scratch after updates and remains tractable when the number of wildcards is small.

The overall procedure is summarized in Algorithm~\ref{alg:dynamic-master}, and the correctness and efficiency of the algorithm are proved in Theorem~\ref{thm:dynamic_wildcard_general}.

\begin{algorithm}[t]
	\caption{Preprocess for Dynamic Pattern Matching with Wildcards}
	\label{alg:dynamic-preprocess}
	\KwIn{Pattern \(\pattern\) of length \(m\), text \(\textstr\) of length \(n\)}
	\KwOut{Modified text \textsf{T'}, occurrence sets \(\mathfrak{R}(\cdot)\), frequent set \(\textsf{F}\)}
	\ForEach{$i \in [1..n]$}{
		\If{\(\textsf{T}_i\) is rare}{
			\(\textsf{T'}_i \gets \#\) \tcp*{replace rare with \(\#\)}
		}
		\Else{\(\textsf{T'}_i \gets \textsf{T}_i\)}
	}
	
	Build \(\mathfrak{R}(c) = \{\,i : \textsf{T}_i=c\,\}\) for each symbol \(c \in \alphabet \cup \{\wildcard\}\) \tcp*{for candidates in Case 1}
	Compute \(\textsf{F} = \{\,c\in\Sigma : |\mathfrak{R}(c)| \ge \tau\,\}\)	\tcp*{Frequent set (for Case 2)}
	Build the \texttt{is\_match} query data structure \tcp*{\(\BigONotation(n)\) preprocessing time}
	Build the \texttt{LCS} query data structure \tcp*{\(\BigONotation(n\log^2 n)\) preprocessing time}
\end{algorithm}

\begin{algorithm}[t]
	\caption{Case 1: \(\pattern\) contains a rare symbol}
	\label{alg:dynamic-case1}
	\KwIn{Pattern \(\pattern\) (with a rare symbol \(c\) at position \(\texttt{pos}\)), text \textsf{T}; sets \(\mathfrak{R}(.)\)}
	\KwOut{Whether \(\pattern\) matches a substring of \(\textstr\)}
	\ForEach{\(i \in \big(\mathfrak{R}(c)\ \cup\ \mathfrak{R}(?)\big)\)}{
		\If{\(\texttt{is\_match}(\pattern,\ \textstr_{\,i-\texttt{pos}+1\ :\ i-\texttt{pos}+m})\)}{
			\Return Match Found
		}
	}
	\Return No Match Found
\end{algorithm}
\begin{algorithm}[t]
	\caption{Case 2: No rare symbols in \(\pattern\) (joint completions)}
	\label{alg:dynamic-case2}
	\KwIn{Pattern \(\pattern\), modified text \textsf{T'}; frequent set \(\textsf{F}\); total wildcards \(k\) across \(\pattern\) and \textsf{T'}}
	\KwOut{Whether \(\pattern\) matches a substring of \(\textstr\)}
	
	\tcp{Enumerate joint completions over \(\textsf{F}\cup\{\#\}\) in Gray-code order}
	\ForEach{completion vector \((a_1,\dots,a_k)\in (\textsf{F}\cup\{\#\})^k\) in Gray-code order}{
		Apply the single-coordinate change implied by \((a_1,\dots,a_k)\) to obtain \((\widetilde{\pattern},\,\widetilde{\textstr})\)\;
		\If{\(\textsf{LCS}(\widetilde{\pattern},\,\widetilde{\textstr}) = m\)}{
			\Return Match Found
		}
	}
	\Return No Match Found
\end{algorithm}

\begin{algorithm}[t]
	\caption{Main Algorithm: Dynamic Pattern Matching with Wildcards}
	\label{alg:dynamic-master}
	
	\SetKwFunction{Preprocess}{Preprocess}           
	\SetKwFunction{Maintain}{Maintain}               
	\SetKwFunction{CaseOne}{CaseOne}                 
	\SetKwFunction{CaseTwo}{CaseTwo}                 
	
	\KwIn{Pattern \(\pattern\) of length \(m\), text \(\textstr\) of length \(n\), stream of updates}
	\KwOut{After each update, whether \(\pattern\) matches a substring of \(\textstr\)}
	
	\Preprocess{\(\pattern, \textstr\)} \tcp*{One-time preprocessing (Alg.~\ref{alg:dynamic-preprocess})}
	
	\ForEach{update on \(\pattern\) or \(\textstr\)}{
		\Maintain{\(\text{current update}\)} \tcp*{Maintain occurrences and \textsf{T'} in \(\BigONotation(\tau)\) (Alg.~\ref{alg:dynamic-maintenance})}
		
		\eIf{\(\pattern\) contains a rare symbol}{
			\tcp{Case 1 (Alg.~\ref{alg:dynamic-case1})}
			\If{\CaseOne{\(\pattern, \textstr\)} \emph{is} \texttt{true}}{
				\Return Match Found
			}
			\Else{
				\Return Match Not Found
			}
		}{
			\tcp{Case 2 (Alg.~\ref{alg:dynamic-case2})}
			\If{\CaseTwo{\(\pattern, \textsf{T'}, \textsf{F}, k\)} \emph{is} \texttt{true}}{
				\Return Match Found
			}
			\Else{
				\Return No Match Found
			}
		}
	}
\end{algorithm}

\begin{theorem}
	\label{thm:dynamic_wildcard_general}
	Algorithm~\ref{alg:dynamic-master} maintains the dynamic pattern matching with wildcards problem with per-query time
	$$
	\Theta\!\Bigl(kn^{\tfrac{k}{k+1}} \,\log^{\tfrac{k+8}{k+1}} n + k^{2}\log n \Bigr).
	$$
	The algorithm always reports a match when one exists, and reports no match with high probability
	when no match exists.
\end{theorem}

\begin{proof}
We adopt a frequency-based approach. Recall that the number of frequent symbols is bounded by
$
\BigONotation\!\left(\frac{n}{\tau}\right).
$
The preprocessing primarily consists of building the data structure for answering LCS queries. 
By Theorem~\ref{thm:lcs_query}, this requires \(\BigONotation(n \log^2 n)\) time. 
In addition, for the matching queries we perform an \(\BigONotation(n)\)-time preprocessing step, as established in Lemma~\ref{thm:is_match_hash_query}.

To answer queries, we distinguish between different cases depending on whether the pattern contains a rare symbol.
	\begin{enumerate}
		\item  \textbf{Pattern contains a rare symbol:}
		If the pattern \(\pattern\) contains a rare symbol \(c\) at some position \(i\) (i.e., \(\pattern_i = c\)), 
		then any valid match in the text must align this occurrence of \(c\) with either the same symbol in \(\textstr\) 
		or with one of the wildcard positions in \(\textstr\).  
		Formally, for a match to occur, we must have \(\textstr_j = c\) or \(\textstr_j = ?\) for some position \(j\) in \(\textstr\).  
		Thus, the candidate positions to ckeck for alignment are $
		\mathfrak{R}(c) \,\cup\, \mathfrak{R}(?).
		$
		To identify potential matches, we scan only these candidate positions. For
		each such \(j\), we verify whether \(\pattern\) matches the substring \(\textstr_{j-i+1 : j-i+m}\). This check is
		performed via the procedure
		$$
		\texttt{is\_match}(\pattern, \textstr_{j-i+1 : j-i+m}),
		$$
		which, by Lemma~\ref{thm:is_match_hash_query}, confirms a match in \(\BigONotation(k \log \lengthT)\) time. 
		Since \(c\) is rare, it appears at most \(\tau\) times in \(\textstr\), so there are at most \(\BigONotation(\tau + k)\) candidate positions, and the total time required for this case is
		$
		\BigONotation\!\left((\tau + k) \cdot k \log \lengthT\right).
		$
			
		\begin{figure}[t]
			\centering
			\scalebox{0.85}{
				\begin{tikzpicture}[
					box/.style={draw, thick, rounded corners, minimum size=0.6cm, font=\ttfamily},
					patternbox/.style={box, fill=blue!10, draw=blue!60},
					wildcardbox/.style={box, fill=yellow!20, draw=orange!60},
					textbox/.style={box, fill=green!10, draw=green!60},
					index/.style={font=\footnotesize, below, text=gray},
					matchrect/.style={draw=blue!80, thick, dashed, rounded corners},
					arrow/.style={-stealth, thick, red!80},
					nonwildmatch/.style={draw=red!80, very thick, rounded corners},
					background/.style={fill=gray!5, draw=gray!20, rounded corners}
					]
					
					\node[anchor=east] at (-0.5, 2) {\(\pattern =\)};
					\foreach \i/\char in {1/a, 2/v, 3/\texttt{?}, 4/b, 5/\texttt{?}, 6/c, 7/d, 8/e, 9/f} {
						\begin{scope}[shift={(\i,2)}]
							\ifnum\i=3 \node[wildcardbox] (p\i) {\char};
							\else \ifnum\i=5 \node[wildcardbox] (p\i) {\char};
							\else \node[patternbox] (p\i) {\char};
							\fi\fi
							\node[index] at (0,-0.4) {\i};
						\end{scope}
					}
					
					\node[anchor=east] at (-0.5, 0) {\(\textstr =\)};
					\foreach \i/\char in {1/x, 2/a, 3/v, 4/t, 5/b, 6/y, 7/c, 8/d, 9/e, 10/f, 11/z} {
						\begin{scope}[shift={(\i,0)}]
							\node[textbox] (t\i) {\char};
							\node[index] at (0,-0.4) {\i};
						\end{scope}
					}
					
					\draw[matchrect] (1.55,-1) rectangle (10.45,1);
					
					\node[nonwildmatch, fit=(p1)(p2)] {}; 
					\node[nonwildmatch, fit=(t2)(t3)] {}; 
					
					\node[nonwildmatch, fit=(p6)(p9)] {}; 
					\node[nonwildmatch, fit=(t7)(t10)] {}; 
					
					\node[nonwildmatch,fit=(p4)] {}; 
					\node[nonwildmatch,fit=(t5)] {}; 
					
					\foreach \pi/\ti in {1/2, 2/3, 4/5, 6/7, 7/8, 8/9, 9/10} {
						\draw[arrow] (p\pi.south) -- (t\ti.north);
					}
					
					\begin{scope}[on background layer]
						\node[background, inner sep=5mm, fit=(current bounding box)] {};
					\end{scope}
					
				\end{tikzpicture}
			}
			\caption{Pattern: \(av?b?cdef\). Text: \(xavtbycdefz\). Red boxes show grouped consecutive non-wildcard matches: "av", "b", and "cdef".}
			\label{fig:rare-merged-fixed-styled}
		\end{figure}
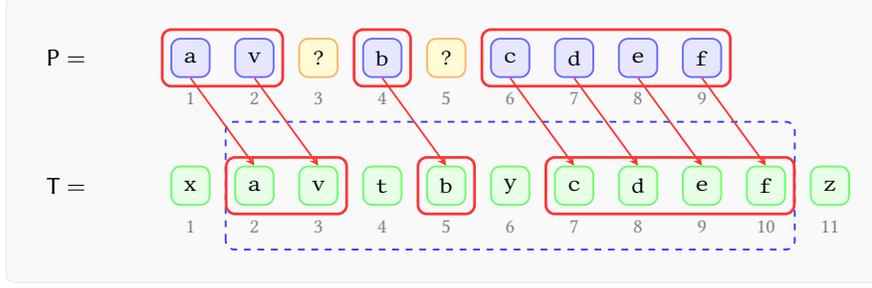
		
\item \noindent\textbf{Pattern contains only frequent symbols and wildcards.}  
In this case, \(\pattern\) contains no rare symbols, so any rare symbol in the text can only be matched by a wildcard in \(\pattern\).  
We therefore cosider a modified text \textsf{T'} by replacing every rare symbol in \textsf{T} with a special placeholder \(\# \notin \Sigma\), while leaving all frequent symbols unchanged.  

Next, we use Gray codes~\cite{knuth2013art} 
to systematically consider all ways to fill the wildcard positions in both \(\pattern\) and \textsf{T'} using symbols from \(\textsf{F} \cup \{\#\}\), where \(\textsf{F}\) is the set of frequent symbols (\(|\textsf{F}| \le \lengthT / \tau\)). We refer to each such choice of symbols as a \emph{joint completion}, because it specifies replacements for wildcards in both strings simultaneously.  

Formally, let \(k\) be the total number of wildcard positions across \(\pattern\) and \textsf{T'}, and index them in a fixed order as \(q_1, \dots, q_k\) (first the wildcards in \(\pattern\), then those in \textsf{T'}). A joint completion is a vector \((a_1, \dots, a_k) \in (\textsf{F} \cup \{\#\})^k\), where each \(q_j\) is replaced by \(a_j\) in the corresponding string. Applying these replacements produces a fully specified pair \((\widetilde{\pattern}, \widetilde{T})\) with no wildcards.  
The number of possible joint completions is therefore bounded by
$$
(|\textsf{F}| + 1)^k \;\le\; \left(\frac{\lengthT}{\tau} + 1\right)^k.
$$
		
Since \(\pattern\) contains no rare symbols, any rare symbol in the text can only match a wildcard in \(\pattern\).  
Therefore, replacing all rare symbols in \(T\) with the placeholder \(\#\) preserves all potential matches.  

When considering text wildcards:  
\begin{itemize}
	\item If a wildcard in \(T\) aligns under a fixed frequent symbol in \(\pattern\), it must take that symbol to match.  
	\item If it aligns under a wildcard in \(\pattern\), any symbol from \(\textsf{F} \cup \{\#\}\) is valid (rare symbols are captured by \(\#\)).  
\end{itemize}

Consequently, any true match between \(\pattern\) and \(T\) corresponds to some completion over \(\textsf{F} \cup \{\#\}\), and conversely, any completed pair \((\widetilde{\pattern}, \widetilde{T})\) that matches represents a genuine match for the original \((\pattern, T)\) once \(\#\) is replaced by the actual rare symbols.

For each completion \((\widetilde{P},\widetilde{T})\), we check whether 
$
\mathrm{LCS}(\widetilde{P}, \widetilde{T}) = \lengthP.
$ 
If so, \(\widetilde{P}\) occurs as a substring of \(\widetilde{T}\), and hence \(\pattern\) matches a substring of \(T\).  
Using the fully dynamic LCS structure (Theorem~\ref{thm:lcs_query}), each single-symbol update and query takes \(\BigONotation(\log^8 \lengthT)\) time.

There are at most \(\left(\tfrac{\lengthT}{\tau}+1\right)^k\) joint completions.  
As shown below, these completions can be enumerated with only \(\BigONotation(1)\) updates per step.  
Since processing each completion requires \(\BigONotation(\log^8 \lengthT)\) time, the overall running time is
$$
\BigONotation\!\left(\left(\tfrac{\lengthT}{\tau}+1\right)^k \cdot \log^8 \lengthT\right).
$$

\end{enumerate}
		
\paragraph*{Gray–code enumeration.}
We enumerate all vectors \((a_1,\dots,a_k)\in[|\textsf{F}|+1]^k\) in a generalized Gray–code order, so that consecutive vectors differ in exactly one coordinate. This ensures that each move between completions changes only a single symbol in \(\widetilde{P}\) or \(\widetilde{T}\).

\begin{itemize}
	\item \textbf{Base case (\(k=1\)):} Use the order \(\{f_1, \dots, f_{|\textsf{F}|}, \#\}\).  
	\item \textbf{Recursive step (\(k>1\)):}  
	Assume we already have a Gray order for \((k-1)\)-tuples.  
	Traverse these tuples in that order, appending the last coordinate alternately in forward (\(f_1, \dots, f_{|\textsf{F}|}, \#\)) and reverse order.  
	Within each block, only the last coordinate changes; between blocks, exactly one of the first \(k-1\) coordinates changes.  
	Thus, each step updates exactly one symbol in \(\widetilde{P}\) or \(\widetilde{T}\).
\end{itemize}

	\begin{example}
		Suppose the pattern is
		$$
		P = \texttt{U ? V ? X},
		$$
		where the two wildcard positions are \(w_1=2\) and \(w_2=4\).  
		Here \(U,V,X\) denote fixed substrings of the pattern that contain no wildcards (possibly empty).  Also, suppose that the text has no wildcards.
		Let the set of frequent symbols be \(\textsf{F}=\{a,b,c\}\), so \(|\textsf{F}|=3\).  
		
		Each concrete candidate pattern is obtained by choosing a pair \((f_1,f_2)\in \textsf{F}^2\), where \(f_1\) is substituted at 
		position \(w_1\) and \(f_2\) at position \(w_2\).  
		For instance, the pair \((a,c)\) yields the candidate \(\texttt{U a V c X}\).  
		
		We generate all pairs in Gray-code order so that each consecutive pair differs in exactly one coordinate. 
		One valid ordering is:
		$$
		(a,a)\;\to\;(a,b)\;\to\;(a,c)\;\to\;(b,c)\;\to\;(b,b)\;\to\;(b,a)\;\to\;(c,a)\;\to\;(c,b)\;\to\;(c,c).
		$$
		
		Mapping these pairs to candidates gives:
		$$
		\begin{array}{c|c|c}
			\text{pair }(f_1,f_2) & \text{changed coordinate} & \text{candidate string}\\ \hline
			(a,a) & - & \texttt{U a V a X}\\
			(a,b) & f_2 & \texttt{U a V b X}\\
			(a,c) & f_2 & \texttt{U a V c X}\\
			(b,c) & f_1 & \texttt{U b V c X}\\
			(b,b) & f_2 & \texttt{U b V b X}\\
			(b,a) & f_2 & \texttt{U b V a X}\\
			(c,a) & f_1 & \texttt{U c V a X}\\
			(c,b) & f_2 & \texttt{U c V b X}\\
			(c,c) & f_2 & \texttt{U c V c X}
		\end{array}
		$$
		
		Since consecutive candidates differ in exactly one wildcard position, transitioning from one row to the next requires 
		only a single-symbol update in the LCS data structure. Thus each step costs \(\BigONotation(\log^8 \lengthT)\).
	\end{example}
		
		For each case, the running time is as follows:  
	$$\text{Case}~(i): \BigONotation\!\left(\tau \cdot k \log \lengthT + k^{2} \log \lengthT \right),
		\qquad \text{Case}~(ii): \BigONotation\!\left(\left(\tfrac{\lengthT}{\tau} + 1\right)^{k} \cdot \log^{8}\lengthT \right).$$
	 Therefore, the combined total time complexity is
	$$
	\BigONotation\!\left(
	\tau \cdot k \log \lengthT 
	\;+\;
	k^{2}  \log \lengthT
	\;+\;
	\left(\tfrac{\lengthT}{\tau} + 1\right)^{k} \cdot \log^{8}\lengthT
	\right).
	$$

	\noindent Finally, by replacing \(\tau = \bigl(n^k \log^{7} n \bigr)^{\tfrac{1}{k+1}}\),  
	we obtain the bound
	$$
	\Theta\!\Bigl(kn^{\tfrac{k}{k+1}} \,\log^{\tfrac{k+8}{k+1}} n + k^{2} \log \lengthT \Bigr).
	$$
	

\begin{algorithm}[t]
	\caption{Dynamic Maintenance of Occurrences and \(\mathsf{T'}\)}
	\label{alg:dynamic-maintenance}
	\KwIn{Update at position \(i\): set \(\mathsf{T}_i \gets c_{\text{new}}\) (old symbol \(c_{\text{old}}\))}
	\KwData{\(\mathsf{T},\,\mathsf{T'}\); occurrence sets \(\{\mathfrak{R}(c)\}_{c\in\Sigma}\); threshold \(\tau\)}
	
	\BlankLine
	remove \(i\) from \(\mathfrak{R}(c_{\text{old}})\)\;
	insert \(i\) into \(\mathfrak{R}(c_{\text{new}})\)\;
	set \(\mathsf{T}_i \gets c_{\text{new}}\)\;
	
	\BlankLine
	
	\BlankLine
	\If{\(|\mathfrak{R}(c_{\text{old}})| < \tau\)}{
		\tcp{symbol \(c_{\text{old}}\) is rare now: set \(\#\) at all its occurrences in \(\mathsf{T'}\)}
		\ForEach{\(j \in \mathfrak{R}(c_{\text{old}})\)}{
			\(\mathsf{T'}_j \gets \#\)\;
		}
	}
	\If{\(|\mathfrak{R}(c_{\text{new}})| \ge \tau\)}{
		\tcp{symbol \(c_{\text{new}}\) is frequent now: restore its real symbol at all its occurrences in \(\mathsf{T'}\)}
		\ForEach{\(j \in \mathfrak{R}(c_{\text{new}})\)}{
			\(\mathsf{T'}_j \gets c_{\text{new}}\)\;
		}
	}
	
	\BlankLine
\end{algorithm}

%
	
\end{proof}

\begin{remark}[Supporting insertions and deletions]
	\label{rem:insdel}
	Our main presentation assumes substitution-only updates, but the framework naturally extends to insertions and deletions by storing the text $\textsf{T}$ in an \emph{implicit-order} balanced \textsf{BST}, specifically an \emph{implicit treap}~\cite{seidel1996randomized}.  
	Each position $i$ of $\textsf{T}$ corresponds to a node $v_i$, and insertions or deletions at index $i$ are implemented via treap splits and merges in \emph{expected} $\BigONotation(\log n)$ time.  
	
	For each symbol $c \in \Sigma$, we maintain its occurrence set $\mathfrak{R}(c)$ as a collection of \emph{pointers to nodes} rather than indices. The current position of an occurrence can be recovered on demand by computing the in-order rank $\mathrm{rank}(v)$ of its node $v$, which also takes \emph{expected} $\BigONotation(\log n)$ time.  	With this setup:
	\begin{itemize}
		\item A \textbf{substitution} at position $i$ changes the label at $v_i$ and moves its pointer between the sets $\mathfrak{R}(c_{\text{old}})$ and $\mathfrak{R}(c_{\text{new}})$.  
		\item An \textbf{insertion} at index $i$ creates a new node $v$ and adds a pointer to $v$ into $\mathfrak{R}(c)$.  
		\item A \textbf{deletion} at index $i$ removes node $v_i$ and deletes its pointer from $\mathfrak{R}(\cdot)$.  
	\end{itemize}
	
	The maintenance of $\textsf{T}'$ and the frequent/rare threshold checks remain unchanged, except that explicit indices are replaced with node-rank queries. Overall, updates incur only an additional \emph{expected} $\BigONotation(\log n)$ cost for treap operations and rank queries, while the $\tau$-bounded work for demotions and promotions is unaffected.  
	
	Moreover, the same treap structure can be used to maintain prefix-hashes and hence support online computation of $H_{b,p}$ in \emph{expected} $\BigONotation(\log n)$ time per update, ensuring that hash-based matching queries remain efficient under insertions and deletions as well.
\end{remark}

	\renewcommand{\eps}{\frac{1}{(\ln \lengthT)^\alpha}}

\section{Sparse Pattern}
\label{sec:sparse}
In this section we study \emph{sparse} patterns, namely patterns \(P\in(\Sigma\cup\{\texttt{?}\})^{m}\) in which only a few positions are non-wildcard and the remaining positions are wildcards. Let \mbox{\(S=\{i\in[m]:\pattern_i\in\Sigma\}\)} be the set of non-wildcard indices and let \(s=|S|\ll m\).  
Such sparse patterns naturally arise in applications such as biological sequence analysis, where only a few positions of a motif are fixed while the rest may vary, or in intrusion detection, where signatures often contain only a few fixed bytes.  
In this section we consider two regimes: (i) the case where the pattern contains at most two non-wildcard symbols, and (ii) the case where the wildcard positions in the pattern are fixed in advance and the text \(\textstr\) contains no wildcards.  
Throughout this section we work in a \emph{substitution-only} dynamic model: updates may change symbols in the text or in the pattern, but no insertions or deletions are allowed, so both lengths \(n\) and \(m\) (and, in the fixed-wildcard regime, the wildcard index set) remain fixed.

\subsection{At Most Two Non-Wildcard Symbols}
\label{sec:two}

In this section, we study a special case of \(\textit{\OurProb}\) under the additional assumption that the pattern \(\pattern\) contains at most two non-wildcard symbols at any time.  
We focus on this substitution-only regime of the problem, and the text \(\textstr\) contains no wildcards.  
But in Remark~\ref{rem:insertiondeletion} and Remark~\ref{rem:wildcards-in-text} we will explain how the algorithm can be extended to handle insertions/deletions in the pattern as well as wildcards in the text.  
Moreover, in this regime we are also able to count the exact number of matches in the same asymptotic complexity.  
For consistency with the other sections we only present the decision version in the main algorithm, but Remark~\ref{rem:counting-matches} explains how counting can also be supported.  
For this setting, we design an algorithm with preprocessing time \(\BigONotation(n^{9/5})\), pattern update time \(\BigONotation(1)\), text update time \(\BigONotation(n^{4/5}\log n)\), and query time \(\BigONotation(n^{4/5})\).

A general observation in sparse patterns is that leading or trailing runs of \(\wildcard\) can be ignored when reasoning about matches.  
These wildcards only shift the range of permissible starting positions but do not affect the relative positions among the concrete symbols.  
Consequently, we can conceptually divide any pattern into three parts: a prefix of wildcards, a block of concrete symbols (the \emph{query block}), and a suffix of wildcards:
$$
\underbrace{\wildcard \dots \wildcard}_{\text{prefix wildcards}} \;
\underbrace{c_1 \dots c_k}_{\text{query block}} \;
\underbrace{\wildcard \dots \wildcard}_{\text{suffix wildcards}}.
$$

After removing the leading and trailing wildcards, let the remaining concrete positions in \(\pattern\) be \(i_1 < i_2 < \dots < i_k\) with symbols \(c_1, c_2, \dots, c_k\).  
A substring \(\textstr_{s : s+\lengthP-1}\) matches \(\pattern\) if and only if
$$
\textstr_{s+i_j-1} = c_j \quad \text{for all } j = 1, \dots, k.
$$  
Equivalently, a match occurs when the concrete symbols appear in the text at positions that preserve the relative offsets \(i_{j+1}-i_j\).  
This perspective allows match queries to be reduced to checking whether these symbols appear together at the required distances within the allowed range of starting positions.

In the special case where the pattern \(\pattern\) contains at most two non-wildcard symbols, the matching task reduces to checking whether two symbols occur in the text at a specific distance from each other.  
This motivates the following formal problem.

\begin{problem}[Dynamic Range-Pair Query]
	\label{prob:range-pair}
	Let \(\RangeAlphabet\) be an alphabet of size \(\RangeAlphabetSize\).
	The input is a dynamic string \(\RangeText \in \RangeAlphabet^{\RangeLen}\).
	The task is to support the following operations:
	\begin{itemize}
		\item \(\UpdateSymbol{i}{c}\): update the symbol at position \(i\) to \(c \in \RangeAlphabet\).
		\item \(\PairQuery{l}{r}{a}{b}{d}\): return the number of indices \(i \in [l,\,r-d-1]\) such that
		\(\RangeText_{i}=a\) and \(\RangeText_{i+d+1}=b\), where \(a,b \in \RangeAlphabet\) and \(d \ge 0\).
	\end{itemize}
\end{problem}

\noindent
This setting is closely related to the classical \emph{Gapped String Indexing} problem~\cite{bille2022gapped}.
In Gapped String Indexing, one preprocesses a static string \(S\) and answers queries given by two patterns \(P_1,P_2\) and a gap interval \([\alpha,\beta]\), asking whether there exist occurrences at positions \(i\) and \(j \ge i\) such that \(j-i \in [\alpha,\beta]\)~\cite{bille2022gapped}.
In the special case where \(|P_1|=|P_2|=1\) and \(\alpha=\beta=d+1\), this reduces to deciding or counting occurrences of two symbols whose positions differ by exactly \(d+1\).
The Dynamic Range-Pair Query problem captures this special case while additionally restricting valid pairs to lie within a query range \([l,r]\), and allowing the underlying text to be updated dynamically.
Viewed this way, the problem can also be interpreted as a dynamic, range-restricted variant of \emph{Shifted Set Intersection} (equivalent to 3SUM indexing in the static setting)~\cite{bille2022gapped,aronov2024general}.

\noindent
To tackle this problem, we use a block decomposition approach.
The text is partitioned into consecutive blocks, and for each block we maintain precomputed counts of symbol pairs whose left endpoint lies inside the block.
Updates are handled lazily: when only a small number of positions in a block are modified, we record these changes explicitly and leave the precomputed counters unchanged; once the number of pending changes exceeds a fixed threshold, the block is rebuilt from scratch.
For a query range $[l,r]$, blocks that lie entirely inside the query interval—referred to as \emph{full blocks}—are answered using their precomputed counters, with a correction for any pending updates stored for those blocks.
Only positions belonging to blocks that intersect the query interval partially—namely, the boundary blocks at the left and right ends of the range—are handled by explicit scanning.

\section*{Range-Pair Algorithm (Details)}\label{app:range-pair}

\noindent
\textbf{Block decomposition and maintained structures.}
We partition the text \(\RangeText\) into
\(\NumBlocks\) consecutive blocks, each of size \(\BlockSize\) (except possibly the last).
Any query interval \([l,r]\) then consists of at most
\(\BigONotation(\NumBlocks)\) full blocks and two partial boundary blocks;
the two boundary blocks together contain at most \(2\BlockSize-2\) symbols, which are processed explicitly.
We denote the first and last indices of block \(i\) by
\(\BlockStart{i}\) and \(\BlockEnd{i}\), respectively.
For a fixed distance \(d\), valid left endpoints lie in \([l,\,r-d-1]\).

We maintain a snapshot string \(\AppliedText\), which records for every position the most recent symbol that has been fully applied to the precomputed tables.  
In addition, for each block \(i\), we maintain the set
$
\Pending_i = \{\, \mathrm{j} \mid \mathrm{j}\ \text{lies in block } i,\ \RangeText_{\mathrm{j}} \neq \AppliedText_{\mathrm{j}} \,\},
$
which stores exactly those indices inside block \(i\) where \(\RangeText\) and \(\AppliedText\) disagree at the current time.  
During updates, changes are first reflected in the corresponding set \(\Pending_i\); once the number of pending modifications in a block exceeds the threshold \(\LazyRebuild\), we recompute all slices of the precomputed structures that touch that block and synchronize \(\AppliedText\) with \(\RangeText\) on its positions.

\begin{fact}\label{fact:pending-set}
	Each set \(\Pending_i\) can be maintained such that insertions and deletions are supported in \(\BigONotation(\log \BlockSize)\) time, and iterating over all elements of \(\Pending_i\) takes \(\BigONotation(|\Pending_i|)\) time, using a balanced binary search tree.
\end{fact}

For handling pairs of symbols across distinct blocks, we maintain counters
\(\PairMatrix_{i,j,a,b,d'}\) for every ordered pair of blocks \(i<j\), 
every pair of symbols \(a,b \in \RangeAlphabet\), 
and every local offset \(d' \in \{0,\ldots,2\BlockSize-2\}\). 
Each entry \(\PairMatrix_{i,j,a,b,d'}\) counts the number of pairs \((u,v)\) such that
\(u\) lies in block \(i\), \(v\) lies in block \(j\), 
\(\AppliedText_{u}=a\), \(\AppliedText_{v}=b\), 
and the local distance is 
$
d' = (\BlockEnd{i}-u) + (v-\BlockStart{j}).
$
(See Figure~\ref{fig:local-distance} for an illustration of the local distance.)
The local distance represents the number of symbols between \(u\) and \(v\), 
excluding all symbols that lie in the full blocks strictly between blocks \(i\) and \(j\). 
The corresponding actual distance is 
$
d = (j-i-1)\BlockSize + d' .
$
These precomputed counters allow all full blocks in a query to be processed without scanning their contents.

To cover pairs inside a block or across two consecutive blocks, we additionally maintain
\(\NearMatrix_{i,a,b,d}\) for every block \(i\), symbols \(a,b \in \RangeAlphabet\),
and distance \(d \in \{0,\ldots,\BlockSize-1\}\).
The value \(\NearMatrix_{i,a,b,d}\) counts pairs \((u,v)\) such that
\(u\) lies in block \(i\), \(v = u+d+1\) lies in block \(i\) or \(i+1\),
and \(\RangeText_u = a\), \(\RangeText_v = b\).
The structure \(\NearMatrix\) is maintained eagerly with respect to the current text \(\RangeText\),
and is responsible for all short-distance pairs with distance strictly less than \(\BlockSize\).

These variables are initialized at the beginning of the algorithm, as described in the preprocessing step (Algorithm~\ref{alg:range-pair-preprocess}).

\begin{figure}[t]
	\centering
	\scalebox{0.65}{
		\begin{tikzpicture}[
			box/.style={draw=black, thick, rounded corners, minimum height=0.6cm, minimum width=0.6cm, font=\ttfamily},
			cell/.style={box, fill=green!20},
			topcell/.style={box, fill=green!20},
			idx/.style={font=\footnotesize, below, text=black},
			label/.style={font=\small},
			measure/.style={<->, thick, draw=black},
			background/.style={fill=gray!5, draw=gray!20, rounded corners},
			textbar/.style={rounded corners=2pt, thick, fill=gray!20, draw=black},
			iblock/.style={rounded corners=2pt, thick, fill=yellow!30, draw=black},
			jblock/.style={rounded corners=2pt, thick, fill=blue!20, draw=black}
			]
			
			\def\TextL{0}  \def\TextR{12}
			\def\BiL{2.0}  \def\BiR{5.0}
			\def\BjL{7.0}  \def\BjR{10.0}
			\def\uTop{3.1}
			\def\vTop{8.4}
			\def\yTop{1.9}
			\def\yMid{1.2}
			\def\yBot{0.3}
			
			\node[anchor=east] at (-0.7,\yTop) {\(\RangeText=\)};
			\draw[textbar] (\TextL+0.2,\yTop-0.35) rectangle (\TextR-0.2,\yTop+0.35);
			\draw[iblock] (\BiL,\yTop-0.33) rectangle (\BiR,\yTop+0.33);
			\draw[jblock] (\BjL,\yTop-0.33) rectangle (\BjR,\yTop+0.33);
			
			\node[label] at ({(\BiL+\BiR)/2}, \yTop+0.60) {$i$-th block};
			\node[label] at ({(\BjL+\BjR)/2}, \yTop+0.60) {$j$-th block};
			
			\node[topcell] (uTopCell) at (\uTop,\yTop) {\(\RangeText_u\)};
			\node[topcell] (vTopCell) at (\vTop,\yTop) {\(\RangeText_v\)};
			
			\foreach \x in {\BiL,\BiR,\BjL,\BjR,\uTop,\vTop} { \draw (\x,\yTop-0.25) -- (\x,\yMid); }
			
			\def\BiLb{2.2}  \def\BiRb{6.2}
			\def\BjLb{6.5}  \def\BjRb{10.5}
			
			\pgfmathsetmacro{\fu}{(\uTop-\BiL)/(\BiR-\BiL)}
			\pgfmathsetmacro{\fv}{(\vTop-\BjL)/(\BjR-\BjL)}
			\coordinate (uBotPos) at ($(\BiLb,\yBot)!\fu!(\BiRb,\yBot)$);
			\coordinate (vBotPos) at ($(\BjLb,\yBot)!\fv!(\BjRb,\yBot)$);
			
			\draw[iblock] (\BiLb,\yBot-0.35) rectangle (\BiRb,\yBot+0.35);
			\draw[jblock] (\BjLb,\yBot-0.35) rectangle (\BjRb,\yBot+0.35);
			
			\node[cell, minimum width=0.7cm] (uSq) at (uBotPos) {};
			\node[cell, minimum width=0.7cm] (vSq) at (vBotPos) {};
			\node[idx] at (uSq.south) {$u$};
			\node[idx] at (vSq.south) {$v$};
			
			\draw (\BiL,\yMid) -- (\BiLb,\yBot+0.35);
			\draw (\BiR,\yMid) -- (\BiRb,\yBot+0.35);
			\draw (\BjL,\yMid) -- (\BjLb,\yBot+0.35);
			\draw (\BjR,\yMid) -- (\BjRb,\yBot+0.35);
			\draw (\uTop,\yMid) -- (uSq.north);
			\draw (\vTop,\yMid) -- (vSq.north);
			
			\coordinate (endiEdge) at (\BiRb,\yBot);
			\coordinate (begjEdge) at (\BjLb,\yBot);
			
			\draw[measure] ([yshift=-18pt]uSq.east) -- node[below, yshift=-2pt] {\(\BlockEnd{i}-u\)} ([yshift=-18pt]endiEdge);
			\draw[measure] ([yshift=-18pt]begjEdge) -- node[below, yshift=-2pt] {\(v-\BlockStart{j}\)} ([yshift=-18pt]vSq.west);
			
			\draw[measure] ([yshift=-44pt]uSq.east) -- node[below, yshift=-2pt]
			{local distance \(d'=(\BlockEnd{i}-u)+(v-\BlockStart{j})\)}
			([yshift=-44pt]vSq.west);
			
			\begin{scope}[on background layer]
				\node[background, inner sep=5mm, fit=(current bounding box)] {};
			\end{scope}
			
		\end{tikzpicture}
	}
	\caption{local distance between \(u \in\) block \(i\) and \(v \in\) block \(j\), equal to
		\((\BlockEnd{i}-u)+(v-\BlockStart{j})\).}
	\label{fig:local-distance}
\end{figure}

\vspace{0.2cm}
\noindent
\textbf{Answering a query.}
For a query \(\PairQuery{l}{r}{a}{b}{d}\), valid left endpoints lie in \([l,\,r-d-1]\).
The computation separates into two parts: (i) the two boundary fragments at the ends of \([l,\,r-d-1]\),
and (ii) the full blocks entirely inside this interval.

\medskip
\noindent\emph{Boundary fragments.}
The two boundary fragments together contain at most \(2\BlockSize-2\) positions.
Each index \(i\) in these fragments is checked directly against \(\RangeText\):
if \(\RangeText_i=a\) and \(\RangeText_{i+d+1}=b\), the counter is incremented.
This part runs in \(\BigONotation(\BlockSize)\) time.

\medskip
\noindent\emph{Full blocks, case \(d<\BlockSize\).}
For every full block \(i\) inside \([l,\,r-d-1]\), the number of pairs with left endpoint in \(i\)
and distance \(d\) is obtained from the precomputed value \(\NearMatrix_{i,a,b,d}\).
Since \(\NearMatrix\) is maintained eagerly with respect to \(\RangeText\),
its entries already reflect the current text and require no correction.
Summation over all full blocks therefore costs \(\BigONotation(\NumBlocks)\), i.e., constant time per block.

\medskip
\noindent\emph{Full blocks, case \(d\ge\BlockSize\).}
Let \(i\) be a full block inside \([l,\,r-d-1]\) and let \(j\) be the block containing
\(\BlockStart{i}+d+1\).
The right endpoint may fall in block \(j\) or block \(j{+}1\), with local offsets, $d' = d-(j-i-1)\BlockSize$, $d'' = d-(j-i)\BlockSize$, respectively; if \((d+1)\) is a multiple of \(\BlockSize\), the \(j{+}1\) term is not needed.
The contribution of block \(i\) is obtained by adding the valid entries
\(\PairMatrix_{i,j,a,b,d'}\) and
\(\PairMatrix_{i,\,j+1,\,a,\,b,\,d''}\).

\medskip
\noindent\emph{Pending-list corrections for \(d\ge\BlockSize\).}
Because \(\PairMatrix\) is maintained lazily, the above sum may overcount or undercount pairs whose
endpoints include positions not yet applied in the precomputed tables.
Corrections are performed only at pending positions, and proceed in two passes:

\begin{itemize}
	\item \textbf{Left-endpoint changes (\(\Pending_i\)).}
	For each \(u\in\Pending_i\), consider the pair \((u,\ u+d+1)\).
	If its right endpoint lies in block \(j\) or \(j{+}1\), compare the symbol under
	\(\AppliedText\) (the value recorded in the precomputed structures) with the symbol under \(\RangeText\).
	If the pair was counted but is no longer a match, decrement; if it was not counted but is now a match, increment.
	\item \textbf{Right-endpoint changes (\(\Pending_j\) and \(\Pending_{j+1}\)).}
	For each right endpoint \(v\) in the pending list of \(j\) (or \(j{+}1\)), set \(u=v-(d+1)\).
	If \(u\) lies in block \(i\) and \(u\notin\Pending_i\) (so the left endpoint is up to date),
	adjust the count in the same manner as above by comparing \(\AppliedText\) with \(\RangeText\).
\end{itemize}

Only pending positions are inspected—entire blocks are never rescanned—so the correction
cost per involved block is \(\BigONotation(\LazyRebuild)\).
Hence, the total cost over all full blocks for \(d\ge\BlockSize\) is
\(\BigONotation(\LazyRebuild\cdot\NumBlocks)\).

\medskip
\noindent\emph{Summary.}
The boundary work costs \(\BigONotation(\BlockSize)\).
For \(d<\BlockSize\), the full-block work costs \(\BigONotation(\NumBlocks)\).
For \(d\ge\BlockSize\), the full-block work costs \(\BigONotation(\LazyRebuild\cdot\NumBlocks)\).
The corresponding pseudocode is given in Algorithm~\ref{alg:range-pair-get}.

\noindent
\textbf{Update operation.}
An update \(\UpdateSymbol{\mathrm{i}}{c}\) modifies the text by setting 
\(\RangeText_{\mathrm{i}} \gets c\), and then updates the maintained data structures so that they again correctly represent the current state of the text.

\medskip
\noindent\emph{Pending set.}
For the block \(b\) containing position \(\mathrm{i}\), the set \(\Pending_b\) records exactly the indices where \(\RangeText\) and \(\AppliedText\) disagree.
After the change, this invariant is restored: the position \(\mathrm{i}\) is inserted into \(\Pending_b\) if \(\RangeText_{\mathrm{i}} \neq \AppliedText_{\mathrm{i}}\), and removed otherwise.

\medskip
\noindent\emph{Short distances (\(d<\BlockSize\)).}
All entries of \(\NearMatrix\) affected by \(\mathrm{i}\) are updated eagerly.  
Since each position can participate in at most \(2\BlockSize\) pairs, this step costs \(O(\BlockSize)\).

\medskip
\noindent\emph{Rebuild policy.}
After every modification, the size of the pending set is checked.  
As soon as \(|\Pending_b| = \LazyRebuild\), a rebuild of block \(b\) is triggered, so the size of \(\Pending_b\) never exceeds \(\LazyRebuild\).

\medskip
\noindent\emph{Rebuild step.}
During a rebuild, \(\AppliedText\) is synchronized with \(\RangeText\) on block \(b\) in \(O(\BlockSize)\), the set \(\Pending_b\) is cleared, and all \(\PairMatrix\) slices involving block \(b\) are recomputed.  
Each block–pair convolution costs \(O(\BlockSize \log \BlockSize \cdot \RangeAlphabetSize^2)\) using FFT (details given later), and since there are \(\NumBlocks\) counterparts, the total cost is \(O(\NumBlocks \cdot \BlockSize \log \BlockSize \cdot \RangeAlphabetSize^2)\).

\medskip
\noindent\emph{Amortized bound.}
Combining the eager updates to \(\NearMatrix\) and the thresholded rebuilds of \(\PairMatrix\), the amortized update time is
$$
O\!\Big(\BlockSize \;+\; \tfrac{\RangeLen \log \BlockSize \cdot \RangeAlphabetSize^2}{\LazyRebuild}\Big).
$$

These steps are implemented in Algorithm~\ref{alg:range-pair-update}.

\paragraph*{Polynomial encoding.}
To compute \(\PairMatrix\), we use polynomial multiplication.
For each block \(i\) and symbol \(a \in \RangeAlphabet\), we define
$$
P_{i,a}(x) = \sum_{u \in \text{block } i} [\AppliedText_{u}=a] \cdot x^{\BlockEnd{i}-u}.
$$
That is, we place a monomial \(x^{\BlockEnd{i}-u}\) for each occurrence of \(a\) in block \(i\), written in right-to-left order.
Similarly, for each block \(j\) and symbol \(b\),
$$
Q_{j,b}(x) = \sum_{v \in \text{block } j} [\AppliedText_{v}=b] \cdot x^{v-\BlockStart{j}}.
$$
Multiplying,
$$
P_{i,a}(x) \cdot Q_{j,b}(x) = \sum_{d'=0}^{2\BlockSize-2} \PairMatrix_{i,j,a,b,d'} \, x^{d'}.
$$
Indeed, every term \(x^{\BlockEnd{i}-u} \cdot x^{v-\BlockStart{j}} = x^{d'}\) corresponds exactly to a pair \((u,v)\) with local distance \(d'\).
Thus the coefficient of \(x^{d'}\) in the product polynomial is precisely \(\PairMatrix_{i,j,a,b,d'}\).
Each convolution takes \(\BigONotation(\BlockSize \log \BlockSize)\) time via FFT
(see e.g.~\cite{CLRS,vzGG}).

\begin{algorithm}[t]
	\caption{Preprocessing for Range-Pair Queries}
	\label{alg:range-pair-preprocess}
	\KwIn{String $\RangeText_{1..\RangeLen}$, block size $B$}
	\KwOut{Tables $\NearMatrix$, $\PairMatrix$, snapshot $\AppliedText$, pending sets $\Pending_i$}
	
	Partition indices into $M=\NumBlocks$ blocks ;
	
	Set $\AppliedText \gets \RangeText$ and $\Pending_i \gets \varnothing$ for all blocks $i$ \tcp*{$\BigONotation(\RangeLen)$}
	
	Build $\NearMatrix_{i,a,b,d}$ for all blocks $i$, symbols $a,b\in\RangeAlphabet$, and distances $0 \leq d < \BlockSize$ \tcp*{$\BigONotation(\RangeLen \cdot \BlockSize)$}
	
	Build $\PairMatrix_{i,j,a,b,d'}$ for all block pairs $i<j$, symbols $a,b\in\RangeAlphabet$, and $0 \leq d' \leq 2\BlockSize-2$ (using FFT on block pairs) \tcp*{$\BigONotation\!\left(\frac{\RangeLen^2}{\BlockSize} \log \BlockSize \cdot \RangeAlphabetSize^2\right)$}
\end{algorithm}

\begin{algorithm}[t]
	\caption{PairQuery($l,r,a,b,d$)}
	\label{alg:range-pair-get}
	\KwIn{Range $[l,r]$, symbols $a,b\in\RangeAlphabet$, distance $d\ge0$}
	\KwOut{Integer \NumberOfMatches = number of $i\in[l,\,r-d-1]$ with $\RangeText_i=a$ and $\RangeText_{i+d+1}=b$}
	
	\NumberOfMatches $\gets 0$\;
	
	\If{$d < \BlockSize$}{
		\ForEach{full block $i$ inside $[l,\,r-d-1]$}{
			\NumberOfMatches $\gets$ \NumberOfMatches $+ \NearMatrix_{i,a,b,d}$\;
		}
	}
	\Else{
		\ForEach{full block $i$ inside $[l,\,r-d-1]$}{
			compute the target blocks $j$ and $j{+}1$ together with their respective offsets $d'$ and $d''$\;
			\NumberOfMatches $\gets$ \NumberOfMatches
			$+ \PairMatrix_{i,j,a,b,d'} + \PairMatrix_{i,\,j+1,\,a,\,b,\,d''}$\;
			apply corrections for $\Pending_{i}\cup\Pending_{j}\cup\Pending_{j+1}$\;
		}
	}
	\ForEach{index $i$ in the two boundary partial blocks of $[l,\,r-d-1]$}{
		\If{$\RangeText_i=a$ and $\RangeText_{i+d+1}=b$}{\NumberOfMatches $\gets$ \NumberOfMatches $+1$}
	}
	
	\Return \NumberOfMatches\;
\end{algorithm}

\begin{algorithm}[t]
	\caption{Update Operation}
	\label{alg:range-pair-update}
	\KwIn{Position $\mathrm{pos}$, new symbol $c\in\RangeAlphabet$}
	\KwOut{Updated data structures}
	
	$\RangeText_{\mathrm{pos}} \gets c$\;
	Let $b$ be the index of the block containing $\mathrm{pos}$\;
	
	\If{$\RangeText_{\mathrm{pos}} \neq \AppliedText_{\mathrm{pos}}$}{
		insert $\mathrm{pos}$ into $\Pending_b$\;
	}
	\Else{
		remove $\mathrm{pos}$ from $\Pending_b$\;
	}
	
	Update all $\NearMatrix_{i,a,b,d}$ entries affected by $\mathrm{pos}$ for $d<\BlockSize$ \tcp*{$\BigONotation(\BlockSize)$}
	
	\If{$|\Pending_b| > \LazyRebuild$}{
		Overwrite $\AppliedText$ in block $b$ with the current $\RangeText$ values \tcp*{$\BigONotation(\BlockSize)$}
		$\Pending_b \gets \varnothing$\;
		Recompute all $\PairMatrix$ slices with one endpoint equal to $b$
		\tcp*{$\BigONotation\!\big(\RangeLen \log \BlockSize \cdot \RangeAlphabetSize^2\big)$}
	}
\end{algorithm}

\begin{theorem}
	\label{thm:range-pair}
	By combining the preprocessing procedure (Algorithm~\ref{alg:range-pair-preprocess}), 
	the update procedure (Algorithm~\ref{alg:range-pair-update}), 
	and the query procedure (Algorithm~\ref{alg:range-pair-get}), 
	the data structure answers the Dynamic Range-Pair Query Problem 
	(Problem~\ref{prob:range-pair}) on \(\RangeText \in \RangeAlphabet^{\RangeLen}\) with:
	\begin{itemize}
		\item Preprocessing time
		$
		\BigONotation\!\left(\tfrac{\RangeLen^2}{\BlockSize} \log \BlockSize \cdot \RangeAlphabetSize^2
		\;+\; \RangeLen \BlockSize\right),
		$
		\item per-query time
		$
		\BigONotation\!\left(\tfrac{\RangeLen}{\BlockSize} \cdot \LazyRebuild \;+\; \BlockSize\right),
		$
		\item amortized update time
		$
		\BigONotation\!\left(\BlockSize \;+\; \tfrac{\RangeLen \log \BlockSize \cdot \RangeAlphabetSize^2}{\LazyRebuild}\right),
		$
		\item space usage
		$
		\BigONotation\!\left(\tfrac{\RangeLen^2}{\BlockSize}\,\RangeAlphabetSize^2\right).
		$
	\end{itemize}
	Here \(\BlockSize\) and \(\LazyRebuild\) are configurable parameters, fixed in advance, that must satisfy $\LazyRebuild \;\leq\; \BlockSize \;\leq\; \RangeLen$.
\end{theorem}

\begin{proof}
	\noindent
	\textbf{Preprocessing time.}
	Block partitioning and initializing \(\AppliedText\) and the sets \(\Pending_i\) take \( \BigONotation(\RangeLen) \).
	Building \(\NearMatrix_{i,a,b,d}\) for all blocks and all \(0 \le d < \BlockSize\) can be done by scanning the text once while updating per-\((a,b)\) counters, for a total of
	$$
	\BigONotation(\RangeLen \BlockSize).
	$$
	For \(\PairMatrix_{i,j,a,b,d'}\), there are \((\frac{\RangeLen}{\BlockSize})^2\) ordered block pairs and \(\RangeAlphabetSize^2\) symbol pairs; each block–pair convolution costs \( \BigONotation(\BlockSize \log \BlockSize) \) via FFT. Hence
	$$
	\BigONotation\!\left(\left(\tfrac{\RangeLen}{\BlockSize}\right)^2 \cdot \BlockSize \log \BlockSize \cdot \RangeAlphabetSize^2\right) = \BigONotation\!\left(\tfrac{\RangeLen^2}{\BlockSize} \log \BlockSize \cdot \RangeAlphabetSize^2\right).
	$$
	Combining the two terms yields the stated preprocessing bound.
	
	\noindent
	\textbf{Query time.}
	Fix a query \(\PairQuery{l}{r}{a}{b}{d}\).
	Valid left endpoints lie in \([l,\,r-d-1]\).
	The two boundary fragments together contain at most \(2\BlockSize-2\) positions, so scanning them against \(\RangeText\) costs \( \BigONotation(\BlockSize) \).
	
	For full blocks inside \([l,\,r-d-1]\), there are \(\BigONotation(\frac{\RangeLen}{\BlockSize})\) blocks to consider.
	When \(d<\BlockSize\), the contribution of each full block \(i\) is read directly from \(\NearMatrix_{i,a,b,d}\), with no correction; summing over all full blocks therefore costs \(\BigONotation(\frac{\RangeLen}{\BlockSize})\).
	This term is subsumed by the later \(\BigONotation(\frac{\RangeLen}{\BlockSize}\cdot \LazyRebuild)\) contribution arising when \(d\ge\BlockSize\), and thus does not affect the stated per-query bound.

	When \(d\ge\BlockSize\), the right endpoint for a left endpoint in block \(i\) lies in block \(j\) or \(j{+}1\), determined by \(j=\) block of \(\BlockStart{i}+d+1\), with local offsets \(d'=d-(j-i-1)\BlockSize\) and \(d''=d-(j-i)\BlockSize\) (the \(j{+}1\) term is unnecessary when \(d{+}1\) is a multiple of \(\BlockSize\)).
	The base contribution from block \(i\) is obtained by summing the valid table entries \(\PairMatrix_{i,j,a,b,d'}\) and \(\PairMatrix_{i,\,j+1,\,a,\,b,\,d''}\).
	Lazy maintenance may render some entries stale; however, corrections inspect only pending positions in the blocks \(i\), \(j\), and (if used) \(j{+}1\).
	By thresholding, each pending set has size at most \(\LazyRebuild\), so the per-block correction is \( \BigONotation(\LazyRebuild) \).
	Summing over \(\BigONotation(\frac{\RangeLen}{\BlockSize})\) full blocks yields
	$$
	\BigONotation\!\left(\frac{\RangeLen}{\BlockSize}\cdot \LazyRebuild\right),
	$$
	and adding the boundary scan gives the stated per-query bound
	\( \BigONotation\!\big( \frac{\RangeLen}{\BlockSize}\cdot \LazyRebuild + \BlockSize \big) \).
	
	\noindent
	\textbf{Amortized update time.}
	An update \(\UpdateSymbol{\mathrm{i}}{c}\) changes \(\RangeText_{\mathrm{i}}\) and restores consistency of the maintained structures.
	The set \(\Pending_b\) for the block \(b\) containing \(\mathrm{i}\) records exactly those indices where \(\RangeText\) and \(\AppliedText\) disagree; insertion/removal keeps this invariant exact.
	All affected \(\NearMatrix\) entries for \(d<\BlockSize\) are updated eagerly; each position participates in at most \(2\BlockSize\) pairs, so this costs \( \BigONotation(\BlockSize) \).
	
	As soon as \(|\Pending_b| = \LazyRebuild\), a rebuild of block \(b\) is triggered:
	\(\AppliedText\) is synchronized with \(\RangeText\) on block \(b\) in \( \BigONotation(\BlockSize) \), \(\Pending_b\) is cleared, and all \(\PairMatrix\) slices with one endpoint equal to \(b\) are recomputed.
	Each block–pair convolution costs \( \BigONotation(\BlockSize\log\BlockSize \cdot \RangeAlphabetSize^2) \), and there are \(\NumBlocks\) counterparts, for a total rebuild cost
	\( \BigONotation(\NumBlocks \cdot \BlockSize\log\BlockSize \cdot \RangeAlphabetSize^2) = \BigONotation(\RangeLen \log\BlockSize \cdot \RangeAlphabetSize^2) \).
	Amortizing this cost over the \(\LazyRebuild\) updates that triggered the rebuild gives
	\( \BigONotation\!\big(\frac{\RangeLen \log\BlockSize \cdot \RangeAlphabetSize^2}{\LazyRebuild}\big) \).
	Adding the eager part yields the stated amortized bound
	\( \BigONotation\!\big(\BlockSize + \frac{\RangeLen \log\BlockSize \cdot \RangeAlphabetSize^2}{\LazyRebuild}\big) \).
	
	\noindent
	\textbf{Space usage.}
	The table \(\NearMatrix\) stores, for each block \(i\) and each \(a,b\in\RangeAlphabet\), the \(d\) range \(0\le d<\BlockSize\), which is
	\( \BigONotation(\RangeLen \cdot \RangeAlphabetSize^2) \).
	The cross–block table \(\PairMatrix\) stores, for every ordered block pair \((i,j)\), every pair \((a,b)\), and every \(d'\in[0,2\BlockSize-2]\), which totals
	$$
	\BigONotation\!\left(\left(\frac{\RangeLen}{\BlockSize}\right)^2 \cdot \RangeAlphabetSize^2 \cdot \BlockSize\right)
	\;=\;
	\BigONotation\!\left(\frac{\RangeLen^2}{\BlockSize}\cdot \RangeAlphabetSize^2\right).
	$$
	This matches the stated space bound.
	
	\noindent
	\textbf{Correctness of the count.}
	Boundary fragments are checked directly in \(\RangeText\), so all pairs with left endpoint in a non–full fragment are counted exactly once.
	For full blocks with \(d<\BlockSize\), \(\NearMatrix_{i,a,b,d}\) is maintained eagerly from \(\RangeText\), so every pair whose left endpoint is in block \(i\) is counted exactly and requires no correction.
	For full blocks with \(d\ge\BlockSize\), the base count from \(\PairMatrix\) includes every pair whose endpoints are both non–pending.
	Pairs with at least one pending endpoint are corrected as follows:  
	(i) if the left endpoint is pending (\(u\in\Pending_i\)), the pair \((u,u+d+1)\) is adjusted according to its current status in \(\RangeText\) (decrement if previously counted but now absent; increment if previously absent but now present);  
	(ii) if the left endpoint is not pending but the right endpoint is pending (in block \(j\) or \(j{+}1\)), the pair is adjusted from the right endpoint, provided the left endpoint lies in block \(i\) and is non–pending (so the base table used the correct left symbol).
	Every pending pair is handled by exactly one of (i) or (ii), and no non–pending pair is altered, so the final total equals the number of pairs satisfying \(\RangeText_i=a\) and \(\RangeText_{i+d+1}=b\) with left endpoint in \([l,\,r-d-1]\).
\end{proof}

Now for the main problem, the algorithm distinguishes the possible cases depending on how many non-wildcard symbols appear in the pattern. 
If the pattern contains a rare symbol, then all its occurrences in the text can be inspected directly, and since at most one further non-wildcard position needs to be checked, the query runs in $\BigONotation(\tau)$ time. 
If the pattern contains only frequent symbols, three subcases arise. 
With no non-wildcards the pattern trivially matches any substring of length $m$ (a constant-time check). 
With exactly one non-wildcard, the query reduces to checking whether that symbol appears in the valid range of positions, which can be done using the set $\mathfrak{R}(c)$ in $\BigONotation(\log n)$ time. 
Finally, with two non-wildcards the problem reduces to asking whether two symbols appear in the text at the required distance, which is exactly the \RangeProblem\ defined earlier. 
To support this reduction we maintain a mapping $\MapSymbol:\alphabet\to\RangeAlphabet$ that assigns all rare symbols to~0 and gives each frequent symbol a distinct nonzero code, so that the mapped text $\RangeText$ serves as input to the Range-Pair structure. 
The procedures for preprocessing, queries, and updates are summarized in Algorithms~\ref{alg:atm2-preprocess}, \ref{alg:atm2-query}, and \ref{alg:atm2-update}. 
By Theorem~\ref{thm:range-pair}, with parameters $\BlockSize=n^{4/5}$, $\LazyRebuild=n^{3/5}$, and $\tau=n^{4/5}$, the resulting bounds match those stated in Theorem~\ref{thm:at-most-two-nonwildcard}.

\begin{algorithm}[t]
  \caption{Preprocess for At-Most-Two Non-Wildcard Symbols}
  \label{alg:atm2-preprocess}
  \KwIn{Text $\textstr \in \alphabet^{\lengthT}$, threshold $\tau$, block size $\BlockSize$}
  \KwOut{Mapped text $\RangeText \in \RangeAlphabet^{\lengthT}$, map $\MapSymbol$; Range-Pair structures on $\RangeText$}

  Classify symbols of $\alphabet$ into frequent/rare using $\tau$ \tcp*{$\BigONotation(\lengthT)$}

  Define $\RangeAlphabet = \{0,1,\dots,M\}$ with $M \coloneqq \lceil 2\lengthT/\tau \rceil$; set $\MapSymbol(x)=0$ for rare $x$, and a nonzero code in $\{1,\dots,M\}$ for frequent $x$ \tcp*{$\BigONotation(\alphabetSize)$}

  Build $\RangeText_t \gets \MapSymbol(\textstr_t)$ for all $t=1..\lengthT$ \tcp*{$\BigONotation (\lengthT)$}

  	Run Range-Pair preprocessing algorithm on $\RangeText$ \tcp*{$\BigONotation\!\left(\tfrac{\RangeLen^2}{\BlockSize} \log \BlockSize \cdot \RangeAlphabetSize^2
  	\;+\; \RangeLen \BlockSize\right)$}
\end{algorithm}

\begin{algorithm}[t]
  \caption{Query for At-Most-Two Non-Wildcard Symbols}
  \label{alg:atm2-query}
  \KwIn{Pattern $\pattern \in (\alphabet \cup \{\wildcard\})^{\lengthP}$, text length $\lengthT$, map $\MapSymbol$}
  \KwOut{Boolean: does $\pattern$ occur in $\textstr$?}

  Let $C$ be the set of non-wildcard positions of $\pattern$; assume $\lengthP \le \lengthT$\;

  \uIf{$|C|=0$}{\Return \texttt{true} \tcp*{$\BigONotation(1)$}}
  \uElseIf{$|C|=1$}{
  	Let $C=\{i\}$ with symbol $a$; let $I \gets [\,i,\ \lengthT-\lengthP+i\,]$\;
  	\Return whether $\mathfrak{R}(a)$ contains an index in $I$ \tcp*{$\BigONotation(\log \lengthT)$}
  }
  \Else(\tcp*[f]{Two non-wildcard symbols}){
  	Let $C=\{i<j\}$ with symbols $a,b$; let $d \gets j-i-1$; let $l \gets i$, $r \gets \lengthT-\lengthP+i$\;
  	\uIf{$a$ is rare}{
  		\ForEach{$p \in \mathfrak{R}(a)$ with $p \in [l,r]$}{
  			\If{$\textstr_{p+d+1} = b$}{\Return \texttt{true}}
  		}
  		\Return \texttt{false} \tcp*{$\BigONotation(\tau)$}
  	}
  	\uElseIf{$b$ is rare}{
  		\ForEach{$q \in \mathfrak{R}(b)$ with $q-d+1 \in [l,r]$}{
  			\If{$\textstr_{q-d+1} = a$}{\Return \texttt{true}}
  		}
  		\Return \texttt{false} \tcp*{$\BigONotation(\tau)$}
  	}
  	\Else{
  		\Return $\,\PairQuery{l}{r}{\MapSymbol(a)}{\MapSymbol(b)}{d} > 0\,$ \tcp*{$\BigONotation\!\left(\tfrac{\lengthT}{\BlockSize}\cdot \LazyRebuild + \BlockSize\right)$}
  	}
  }
\end{algorithm}

\begin{algorithm}[t]
	\caption{Update operations for At-Most-Two Non-Wildcard Symbols}
	\label{alg:atm2-update}
	
	\textbf{Pattern update $(i,x)$:}\\
	\Indp
	$\pattern_i \gets x$ \tcp*{$\BigONotation(1)$}
	Maintain the set $C$ of up to two non-wildcard positions of $\pattern$ \tcp*{$\BigONotation(1)$}
	\Indm
	
	\BlankLine
	
	\textbf{Text update $(i,x)$:}\\
	\Indp
	$\textstr_i \gets x$\;
	maintain the frequency class of $x$ as per threshold $\tau$ (promotion/demotion) \tcp*{amort.~$\BigONotation(1)$ per symbol}
	Perform the \emph{range-pair update} on index $i$ with new symbol $\MapSymbol(x)$ (Algorithm~\ref{alg:range-pair-update}) \tcp*{amort.~$\BigONotation\!\left(\BlockSize + \frac{\lengthT \log \BlockSize \cdot \RangeAlphabetSize^2}{\LazyRebuild}\right)$}
	\Indm
\end{algorithm}

\begin{theorem}
	\label{thm:at-most-two-nonwildcard}
	Combining the preprocessing in Algorithm~\ref{alg:atm2-preprocess}, the query procedure in Algorithm~\ref{alg:atm2-query}, and the text/pattern updates in Algorithm~\ref{alg:atm2-update}, the data structure answers the at-most-two non-wildcard dynamic wildcard matching variant with the following guarantees (for $\textstr\in\alphabet^{\lengthT}$):
	\begin{itemize}
		\item preprocessing time $\BigONotation(\lengthT^{\frac{9}{5}})$,
		\item pattern-update time $\BigONotation(1)$,
		\item text-update time $\BigONotation(\lengthT^{\frac{4}{5}}\log \lengthT)$ (amortized),
		\item per-query time $\BigONotation(\lengthT^{\frac{4}{5}})$.
	\end{itemize}
\end{theorem}

\begin{proof}
	We follow a similar approach to Theorem~\ref{thm:dynamic_wildcard_general}, classifying symbols into frequent and rare, but with a modified choice of \(\tau\). Since the number of frequent symbols is at most \(\frac{\lengthT}{\tau},\) we leverage this property for the optimization.
	
	\medskip
	
	\textbf{Case 1: \(\pattern\) contains a rare symbol.}
	If \(\pattern\) contains a rare symbol, we iterate over all its occurrences in \(\textstr\). For each occurrence at position \(j\), a check is performed to determine whether \(\pattern\) matches the substring \(\textstr_{j - i + 1, j - i + \lengthP}\) by verifying the symbol at the other non-wildcard position (if any) in \(\BigONotation(1)\) time, since at most one other position requires checking. The total time complexity for this case is \(\BigONotation(\tau)\).
	
	\medskip
	
	\textbf{Case 2: \(\pattern\) contains only frequent symbols and wildcards.}	
	The pattern \(\pattern\) falls into one of three categories, each handled independently:
	
	\begin{enumerate}[label=\roman*)]
		\item \textbf{No Non-Wildcard Symbols.}
		In this case, \(\pattern\) consists entirely of wildcards and trivially matches any substring of \(\textstr\) of length \(\lengthP\). We only need to verify that \(\lengthP \leq \lengthT\), which is a constant-time check.
		
		\item \textbf{One Non-Wildcard Symbol.}
		Assume \(\pattern_{i} = c\), where \(c\) is a non-wildcard symbol, and all other positions in \(\pattern\) are wildcards.  
		We aim to find a position \(j\) in \(\textstr\) such that \(\textstr_{j} = c\) and there are at least \(i - 1\) symbols to the left and \(\lengthP - i\) symbols to the right of \(j\).  
		Equivalently, we seek any index \(j\) such that
		$$
		i \leq j \leq \lengthT - \lengthP + i
		\quad \text{and} \quad
		\textstr_{j} = c.
		$$
		
		We derive the answer using the set \(\mathfrak{R}(c)\).  
		To do so, we look for the smallest index in \(\mathfrak{R}(c)\) that is greater than or equal to \(i\).  
		If this index is at most \(\lengthT - \lengthP + i\), then it constitutes a valid match; otherwise, no match exists.
		
		\item \textbf{Two Non-Wildcard Symbols.}
		Suppose \(\pattern_{i} = c\), \(\pattern_{j} = c'\), and all other positions in \(\pattern\) are wildcards.  
		A match of \(\pattern\) in \(\textstr\) corresponds to positions \(i', j'\) in \(\textstr\) such that:
		\begin{itemize}
			\item \(\textstr_{i'} = c\), \(\textstr_{j'} = c'\),
			\item \(j - i - 1 \;=\; j' - i' - 1\) (equal gap),
			\item \(i \leq i'\) and \(j' \leq \lengthT - \lengthP + j\).
		\end{itemize}
		Let \(d \coloneqq j - i - 1\) be the \emph{gap} (number of symbols strictly between the two concrete positions), so the constraint is \(j' = i' + d + 1\).
		Equivalently, the query is:
		$$
		\text{Does there exist } i' \in [l,r] \text{ such that } \textstr_{i'} = c \text{ and } \textstr_{i'+d+1} = c' \, ?
		$$
		This is exactly a range-pair instance as in Problem~\ref{prob:range-pair}.

		\paragraph*{Reduction to the \(\textit{\RangeProblem}\) structure.}
		We build the input \(\RangeText\) and alphabet \(\RangeAlphabet\) for the Range-Pair problem as follows.
		Let \(\tau>0\) be the frequency threshold used to classify \emph{frequent} vs.\ \emph{rare} symbols in the text.
		We maintain a mapping \(\MapSymbol : \alphabet \to \RangeAlphabet\) with:
		\begin{align*}
			\RangeAlphabet &\;=\; \{0,1,2,\dots, M\}, \qquad M \coloneqq \big\lceil 2\lengthT/\tau \big\rceil,\\
			\MapSymbol(x) &= 
			\begin{cases}
				0, & \text{if } x \text{ is \emph{rare} (has}\; |\mathfrak{R}(x)| \le \tau\text{)},\\[2pt]
				z \in \{1,\dots,M\}, & \text{if } x \text{ is \emph{frequent}},
			\end{cases}
		\end{align*}
		so \(|\RangeAlphabet| = \RangeAlphabetSize \le M{+}1 = \BigONotation(\lengthT/\tau)\).
		We then define the Range-Pair text by
		$$
		\RangeText_{t} \;\coloneqq\; \MapSymbol(\textstr_{t}) \qquad \text{for all } t = 1,\dots,\lengthT.
		$$
		
		\paragraph*{Query translation.}
		For a two-symbol pattern \((c,c')\) with gap \(d\) and starting range \([l,r]\), we query the Range-Pair structure with
		$$
		\PairQuery{l}{r}{\MapSymbol(c)}{\MapSymbol(c')}{d}.
		$$
		
		\paragraph*{Mapping invariants and updates.}
		We maintain two simple invariants:
		\begin{enumerate}[label=(\roman*)]
			\item Every \emph{frequent} symbol maps to a \emph{nonzero} code: \(\MapSymbol(x)\in\{1,\dots,M\}\).
			\item Every \emph{rare} symbol maps to \(0\): \(\MapSymbol(x)=0\).
		\end{enumerate}
		On a text update \(\textstr_{\mathrm{pos}} \gets x\), we update \(\RangeText_{\mathrm{pos}} \gets \MapSymbol(x)\) and maintain the invariants:
		\begin{itemize}
			\item \emph{Promotion.} If a rare symbol \(x\) exceeds frequency \(\tau\), assign it a fresh nonzero code in \(\{1,\dots,M\}\) (freeing one if needed by demoting a symbol below) and set \(\MapSymbol(x)\) to that code. This costs \(\BigONotation(\tau)\) work to retag occurrences, but such promotions happen only after \(\Omega(\tau)\) increments, so the amortized cost is \(\BigONotation(1)\) per update per symbol.
			\item \emph{Demotion.} If a frequent symbol \(y\) drops to \(\le \tau/2\) occurrences, we set \(\MapSymbol(y)\gets 0\) and convert its occurrences in \(\RangeText\) to \(0\) in \(\BigONotation(\tau)\) time. This can happen only after \(\Omega(\tau)\) decrements, hence \(\BigONotation(1)\) amortized per symbol.
			\item Otherwise, simply write \(\RangeText_{\mathrm{pos}} \gets \MapSymbol(x)\) (no remapping).
		\end{itemize}
		With the slack \(M = \lceil 2\lengthT/\tau \rceil\), there is always a nonzero code available for promotions (if the pool is exhausted, some currently coded symbol must have \(\le \tau/2\) occurrences and can be demoted), keeping \(\RangeAlphabetSize = \BigONotation(\lengthT/\tau)\) at all times.
		
		\paragraph*{Complexities via the Range-Pair theorem.}
		By Theorem~\ref{thm:range-pair}, operating on \(\RangeText \in \RangeAlphabet^{\lengthT}\) with \(\RangeAlphabetSize \le 2\lceil \lengthT/\tau \rceil{+}1\), the Range-Pair data structure supports:
		\begin{itemize}
			\item Preprocessing in 
			\(
			\BigONotation\big(\big(\tfrac{\lengthT}{\BlockSize}\big)^2 \BlockSize \log \BlockSize \cdot \RangeAlphabetSize^2 \;+\; \lengthT \BlockSize\big),
			\)
			\item Per query \(\BigONotation\!\big(\tfrac{\lengthT}{\BlockSize} \cdot \LazyRebuild + \BlockSize\big)\) for large gaps (and \(\BigONotation(\tfrac{\lengthT}{\BlockSize} + \BlockSize)\) for small gaps),
			\item Amortized update \(\BigONotation\!\big(\tfrac{\lengthT \log \BlockSize \cdot \RangeAlphabetSize^2}{\LazyRebuild} + \BlockSize\big)\).
		\end{itemize}
	\end{enumerate}
	
	\paragraph*{Final Choice.}
	In the rare symbol case (Case 1), queries take \(\BigONotation(\tau)\). For the frequent symbol case (Case 2), we fix parameters:
	$$
	\BlockSize = \lengthT^{\frac{4}{5}}, \quad \LazyRebuild = \lengthT^{\frac{3}{5}}, \quad \tau = \lengthT^{\frac{4}{5}}.
	$$
	With this choice:
	\begin{itemize}
		\item The amortized update time becomes:
		$$
		\BigONotation\left(\frac{\lengthT \log \lengthT \cdot (\frac{\lengthT}{\tau} + 1)^2}{\LazyRebuild}\right) = \BigONotation(\lengthT^{\frac{2}{5}} (\lengthT^{\frac{1}{5}} + 1)^2 \log \lengthT) = \BigONotation(\lengthT^{\frac{4}{5}} \log \lengthT).
		$$
		\item The query time is:
		$$
		\BigONotation\left(\frac{\lengthT}{\BlockSize} \cdot \LazyRebuild + \BlockSize + \tau\right) = \BigONotation(\lengthT^{\frac{4}{5}}).
		$$
		\item The preprocessing time is:
		$$
		\BigONotation\left(\left(\frac{\lengthT}{\BlockSize}\right)^2 \cdot \BlockSize \log \BlockSize \cdot (\frac{\lengthT}{\tau} + 1)^2\;+\; \lengthT \BlockSize\right) = \BigONotation(\lengthT^{\frac{6}{5}} (\lengthT^{\frac{1}{5}} + 1)^2 \log \lengthT\;+\; \lengthT \cdot \lengthT^{\frac{4}{5}}) = \BigONotation(\lengthT^{\frac{9}{5}}).
		$$
	\end{itemize}
	
\end{proof}

\begin{remark}[Insertions and deletions in the pattern]
	\label{rem:insertiondeletion}
	The algorithm also extends to handle insertions and deletions in the pattern \(\pattern\).  
	Such operations only change the length of \(\pattern\), and therefore affect only the range \([l,r]\) over which we search for answers.  
	The underlying data structures and their update/query times remain unchanged.
\end{remark}

\begin{remark}[Wildcards in the text]
	\label{rem:wildcards-in-text}
	Our approach also handles wildcard symbols in the text \(\textstr\).  
	This is achieved by mapping the wildcard `\(\wildcard\)` to a dedicated symbol in the \RangeProblem\ alphabet (e.g., code $1$), and ensuring that no other symbol maps to this code.  
	Unlike ordinary symbols, wildcards in the text are treated specially: they are exempt from the frequent/rare classification.  
	During a query \(\PairQuery{l}{r}{a}{b}{d}\), we additionally evaluate the three variants where either or both endpoints are mapped wildcards, namely
	\(\PairQuery{l}{r}{\MapSymbol(\wildcard)}{\MapSymbol(b)}{d}\),
	\(\PairQuery{l}{r}{\MapSymbol(a)}{\MapSymbol(\wildcard)}{d}\),
	and
	\(\PairQuery{l}{r}{\MapSymbol(\wildcard)}{\MapSymbol(\wildcard)}{d}\).
	This allows wildcard positions in the text to be treated consistently as matching any symbol.

	\noindent\emph{Rare endpoints.}
	If, in any of the four queries above (the base query plus the three wildcard variants), an endpoint symbol is a \emph{rare} symbol, we replace the corresponding \(\PairQuery\){.}{.}{.}{.}{.} call by a direct scan over the occurrences of that rare symbol in the valid range.
	Each such scan costs \(\BigONotation(\tau)\) since \(|\mathfrak{R}(a)|\le\tau\) for rare \(a\).
\end{remark}

\begin{remark}[Counting matches]
	\label{rem:counting-matches}
	In contrast to our other settings, for the at-most-two non-wildcard case we can compute not only existence but also the exact number of matches within the same asymptotic bounds.  
	If the pattern contains a rare symbol, we already enumerate all of its occurrences in \(\textstr\), so counting adds no extra cost.  
	If the pattern has no non-wildcard symbols, every position is a match.  
	If it has exactly one non-wildcard symbol, the answer is the number of occurrences of that symbol in the relevant range, which can be maintained in \(\BigONotation(\log \lengthT)\) time using a Fenwick tree~\cite{fenwick1994}.  
	Finally, if it has two non-wildcard symbols, the query reduces to the Range-Pair problem, for which we have already shown how to compute the exact number of matching pairs.  
	Thus, counting queries can be answered within the same time complexity as decision queries.
\end{remark}

\subsection{Fixed Wildcard Positions in Pattern}

In this subsection we study the dynamic pattern matching problem with fixed wildcards: all $k$ wildcard symbols appear only in the pattern $\pattern$, their index set $W=\{i\in[m]:\pattern_i=\texttt{?}\}$ is known in advance and immutable, and the text $\textstr$ contains no wildcards. Both $\pattern$ and $\textstr$ evolve only via substitution updates, and insertions and deletions are not allowed; therefore updates to $\pattern$ are confined to the solid coordinates $[m]\setminus W$, while updates to $\textstr$ may change any $\textstr_i\in\Sigma$. We parameterize by the number of non-wildcard positions $\nonWild=|[m]\setminus W|$ and target the sparse regime $\nonWild\ll m$. To handle this setting, we adopt a masked polynomial rolling hash that ignores the fixed wildcard indices: equivalently, we view $\pattern$ as a concatenation of solid blocks separated by wildcards and compute the hash only over these blocks, enabling comparison against substrings of $\textstr$ under substitutions while treating wildcard coordinates as don’t-cares; this leads to the following definition.

\begin{definition}
	\label{def:h_prime_new}
	Let \(S\) be a string of length \(\ell\), and let \(\pattern\) be a pattern with \(k\) wildcard positions 
	\(w_1 < w_2 < \dots < w_k\), where each \(w_i\) is the index of the \(i\)-th wildcard in \(\pattern\). We define
	$$
	\hashPrimeFunc{S} = 
	\hashFunc{\, S_{1:w_1-1} \concat S_{w_1+1:w_2-1} \concat \dots \concat S_{w_k+1:\ell} \,},
	$$
	where \(\hashFunc{\cdot}\) is the polynomial rolling hash from Definition~\ref{def:hash_function}.
\end{definition}

The value \(\hashPrimeFunc{S}\) can be computed in \(\BigONotation(k)\) time by concatenating the precomputed hash values of the 
\(k+1\) non-wildcard intervals, using Lemma~\ref{lem:concat_hash_efficiency}. Moreover, after each update in \(S\), 
the value of \(\hashPrimeFunc{S}\) can be updated in \(\BigONotation(1)\) time.  
Let us now define an auxiliary vector \(A\), where
$$
A_i = \hashPrimeFunc{\textstr_{i:i+m-1}} \qquad \text{for } 1 \leq i \leq n - m + 1,
$$
so that \(A_i\) stores the hash of the length-\(m\) substring of \(\textstr\) starting at position \(i\), with 
the fixed wildcards of \(\pattern\) ignored.

\begin{lemma}
	\label{lem:array_computation}
	If the pattern \(\pattern\) decomposes into \(l\) non-empty intervals separated by its wildcard positions, then the 
	vector \(A\) can be constructed in \(\BigONotation(n \cdot l)\) time.
\end{lemma}
\begin{proof}
	The vector \(A\) has \(n - m + 1\) entries, where each entry is of the form 
	\(A_i = \hashPrimeFunc{\textstr_{i:i+m-1}}\). Since \(\pattern\) decomposes into \(l\) non-empty intervals, the substring 
	\(\textstr_{i:i+m-1}\) is split into at most \(l\) corresponding segments. The hash of each segment can be computed 
	in \(\BigONotation(1)\) time using the polynomial rolling hash from Definition~\ref{def:hash_function}. Concatenating 
	these \(l\) segment hashes into \(\hashPrimeFunc{\textstr_{i:i+m-1}}\) requires \(\BigONotation(l)\) time by 
	Lemma~\ref{lem:concat_hash_efficiency}. Therefore, computing one entry of \(A\) costs \(\BigONotation(l)\) time, and 
	computing all \(\BigONotation(n)\) entries requires \(\BigONotation(n l)\) time in total.
\end{proof}

Together, Definition~\ref{def:h_prime_new} and Lemma~\ref{lem:array_computation} provide the key ingredients for 
efficient preprocessing in the sparse-pattern setting. The modified hash \(\hashPrimeFunc{S}\) ensures that 
wildcards are ignored, while the vector \(A\) stores precomputed values for all substrings of \(\textstr\) 
that need to be compared against \(\pattern\). The lemma guarantees that this preprocessing can be done in 
\(\BigONotation(n \cdot l)\) time, where \(l\) is the number of non-empty intervals in \(\pattern\). 
Algorithm~\ref{alg:sparse-pattern} then presents the full method for handling updates. 
Intuitively, when the pattern has only a few fixed symbols and the rest are wildcards, we can ignore 
wildcard positions in all computations. Let \(C=\{c_1<\cdots<c_{\nonWild}\}\) be the indices of the 
non-wildcard positions in \(\pattern\). For every window \(\textstr_{i:i+\lengthP-1}\) of the text, we 
compute \(A_i = \hashPrimeFunc{\textstr_{i:i+\lengthP-1}}\) based only on the symbols aligned with \(C\). 
We also maintain the single value \(\hashPrimeFunc{\pattern}\). A match exists iff 
\(\hashPrimeFunc{\pattern}\) appears among \(\{A_1,\dots,A_{\lengthT-\lengthP+1}\}\); this can be checked efficiently with a membership query, e.g., in \(\BigONotation(\log n)\) time using a balanced \textsf{BST}.

If the pattern changes at a non-wildcard position, we simply recompute \(\hashPrimeFunc{\pattern}\) and test 
membership once. If the text changes at position \(j\), only those windows whose alignment uses \(j\) at one 
of the non-wildcard positions can be affected. Equivalently, for each \(i\in C\) with \(j \ge i\), only 
\(A_{\,j-i+1}\) must be recomputed. Thus a single text update triggers at most \(|C|=\nonWild\) 
window-hash updates.

\begin{example}
	Let \(\pattern=\texttt{?b??a}\) with \(\lengthP=5\) and non-wildcard indices \(C=\{2,5\}\).
	Let \(\textstr=\texttt{cabyzacde}\) with \(\lengthT=9\).
	Consider the alignment starting at \(i=2\) (i.e., substring \(\textstr_{2:6}=\texttt{abyza}\)).
	The two non-wildcard checks are:
	$$
	\pattern_2{=}\texttt{b}\ \stackrel{?}{=}\ \textstr_{2+2-1}{=}\textstr_3{=}\texttt{b}, 
	\qquad
	\pattern_5{=}\texttt{a}\ \stackrel{?}{=}\ \textstr_{2+5-1}{=}\textstr_6{=}\texttt{a},
	$$
	so this alignment matches.
	Accordingly, \(A_2 = \hashPrimeFunc{\textstr_{2:6}}\) equals \(\hashPrimeFunc{\pattern}\).
	
	Now suppose we update the text at position \(j=6\), changing \(\textstr_6\) from \(\texttt{a}\) to \(\texttt{x}\).
	Only the window hashes whose aligned non-wildcard positions use \(j\) can change:
	for \(i=2\) we must recompute \(A_{j-c+1}=A_5\); for \(i=5\) we must recompute \(A_{j-i+1}=A_2\).
	In particular, \(A_2\) will no longer equal \(\hashPrimeFunc{\pattern}\) since \(\textstr_6\) no longer matches \(\pattern_5{=}\texttt{a}\). \emph{See Figure~\ref{fig:sparse-example} for a visualization of the alignment at \(i=2\), the update at \(j=6\), and the affected window hashes.}
\end{example}
\begin{figure}[t]
	\centering
	\scalebox{0.65}{
		\begin{tikzpicture}[
			box/.style={draw, thick, rounded corners, minimum width=0.7cm, minimum height=0.7cm, font=\ttfamily},
			pat/.style={box, fill=blue!8, draw=blue!60},
			wild/.style={box, fill=yellow!25, draw=orange!70!black},
			txt/.style={box, fill=green!8, draw=green!60},
			idx/.style={font=\footnotesize, below, text=gray!70},
			matchrect/.style={draw=magenta!80, thick, dashed, rounded corners},
			update/.style={draw=red!70, very thick, rounded corners},
			arrow/.style={-stealth, thick},
			note/.style={font=\footnotesize},
			background/.style={fill=gray!5, draw=gray!20, rounded corners}
			]
			
			\node[anchor=east] at (-0.7,1.8) {$\pattern=$};
			\foreach \i/\ch/\sty in {1/?/wild,2/b/pat,3/?/wild,4/?/wild,5/a/pat}{
				\begin{scope}[shift={(\i,1.8)}]
					\node[\sty] (p\i) {\ch};
					\node[idx] at (0,-0.45) {\i};
				\end{scope}
			}
			
			\node[anchor=east] at (-0.7,0) {$\textstr=$};
			\foreach \i/\ch in {1/c,2/a,3/b,4/y,5/z,6/a,7/c,8/d,9/e}{
				\begin{scope}[shift={(\i,0)}]
					\node[txt] (t\i) {\ch};
					\node[idx] at (0,-0.45) {\i};
				\end{scope}
			}
			
			\draw[matchrect] (1.55,-0.9) rectangle (6.45,0.9);
			
			\draw[arrow,magenta!80] (p2.south) .. controls +(0,-0.7) and +(0,0.7) .. (t3.north);
			\draw[arrow,magenta!80] (p5.south) .. controls +(0,-1.1) and +(0,1.1) .. (t6.north);
			
			\draw[update] (5.55,-0.35) rectangle (6.45,0.35);
			\node[note,red!70] at (6,-1.2) {update at $j=6$};
			
			\node[note,align=left] at (8.9,1.7)
			{\begin{tabular}{@{}l@{}}
					recompute: \\
					$c=2 \Rightarrow A_{6-2+1}=A_5$ \\
					$c=5 \Rightarrow A_{6-5+1}=A_2$
			\end{tabular}};
			
			\begin{scope}[on background layer]
				\node[background, inner sep=5mm, fit=(current bounding box)] {};
			\end{scope}
			
		\end{tikzpicture}
	}
	\caption{Sparse-pattern example for \(\pattern=\texttt{?b??a}\) and \(\textstr=\texttt{cabyzacde}\).
		The dashed magenta box highlights the alignment at \(i=2\) (substring \(\textstr_{2:6}=\texttt{abyza}\)); arrows show the two non-wildcard checks at pattern positions \(2\) and \(5\).
		The red rectangle marks the text update at \(j=6\) (\(\texttt{a}\!\to\!\texttt{x}\)), which invalidates the match at \(i=2\) and forces recomputation of \(A_5\) and \(A_2\) (as indicated on the right).}
	\label{fig:sparse-example}
\end{figure}
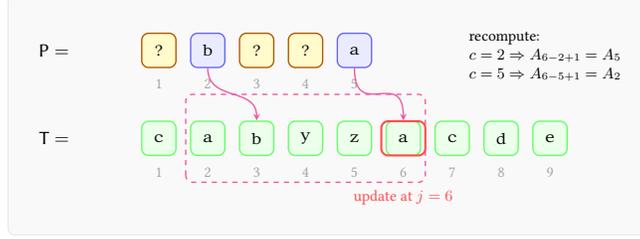

\begin{algorithm}[t]
	\caption{Updates with Preprocessed Data (Sparse Pattern Matching)}
	\label{alg:sparse-pattern}
	\SetKwInOut{Input}{Input}
	\SetKwInOut{Output}{Output}
	
	\Input{Current \(\textstr,\pattern\); preprocessed \(A\), \(C\); lengths \(\lengthT,\lengthP\)}
	\Output{After each update, report whether \(\pattern\) matches a substring of \(\textstr\)}
	
	Apply the substitution in the pattern\;
	\If{\(\hashPrimeFunc{\pattern} \in \{A_1,\dots,A_{\lengthT-\lengthP+1}\}\)}{
		\Return Match Found
	}
	\Return No Match Found

	\BlankLine
	\tcp{(B) Text update at position \(j\): \(\textstr_j \gets c_{\mathrm{new}}\)}
	Apply the substitution in the text\;
	\ForEach{\(i \in C\) \KwSty{with} \(j \ge i\)}{
		\tcp{Recompute only the affected window \(A_{j-i+1}\) if it is in range}
		\If{\(1 \le j-i+1 \le \lengthT-\lengthP+1\)}{
			\(A_{\,j-i+1} \gets \hashPrimeFunc{\textstr_{\,j-i+1 :\, j-i+\lengthP}}\)\;
		}
	}
	\If{\(\hashPrimeFunc{\pattern} \in \{A_1,\dots,A_{\lengthT-\lengthP+1}\}\)}{
		\Return Match Found
	}
	\Return No Match Found
	
\end{algorithm}

\begin{theorem}
	\label{thm:sparse-pattern}
	Algorithm~\ref{alg:sparse-pattern} maintains the dynamic pattern matching problem with fixed wildcards 
	for a pattern \(\pattern\) of length \(m\) containing at most \(\nonWild\) non-wildcard positions.  
	It requires preprocessing time \(\BigONotation(n \cdot \nonWild)\) and supports updates to either \(\pattern\) or \(\textstr\) 
	in \(\BigONotation(\nonWild + \log n)\) time per update.  
	The algorithm always detects a match correctly when one exists.  
	When no match exists, it reports ``no match'' with high probability.
\end{theorem}

\begin{proof}
	We first construct the vector \(A\), where each entry \(A_i\) stores the hash of the substring 
	\(\textstr_{i:i+m-1}\) restricted to the non-wildcard positions of \(\pattern\).  
	Since \(\pattern\) contains at most \(\nonWild\) non-wildcards, all values of \(A\) can be computed in 
	\(\BigONotation(n \cdot \nonWild)\) time (Lemma~\ref{lem:array_computation}).  
	The initial hash value \(\hashPrimeFunc{\pattern}\) is computed in \(\BigONotation(\nonWild)\) time.
	
	\smallskip
	\noindent\textbf{Update on \(\pattern\).}  
	If a symbol of \(\pattern\) is modified, only its contribution to \(\hashPrimeFunc{\pattern}\) changes.  
	By the rolling-hash property (Definition~\ref{def:h_prime_new}), this modification can be updated in constant time.
	
	\smallskip
	\noindent\textbf{Update on \(\textstr\).}  
	Suppose that symbol \(\textstr_j\) of the text is modified.  
	Each non-wildcard position \(q\) of \(\pattern\) aligns with \(\textstr_j\) in exactly one substring \(\textstr_{i:i+m-1}\), namely where 
	\(i = j - q + 1\).  
	Consequently, at most \(\nonWild\) entries of \(A\) are affected.  
	Each affected entry can be recomputed in constant time, so the update to \(\textstr\) requires \(\BigONotation(\nonWild)\) time.
	
	\smallskip
	\noindent\textbf{Checking for matches.}  
	To test whether a match exists, the algorithm checks if
	$$
	\hashPrimeFunc{\pattern} \in \{ A_1, A_2, \dots, A_{n-m+1} \}.
	$$
	If equality holds for some \(i\), then \(\pattern\) and \(\textstr_{i:i+m-1}\) coincide on all non-wildcard positions, which by definition means that \(\pattern\) matches \(\textstr\) at position \(i\).
	This can be checked with a membership query, i.e., verifying whether 
	\(\hashPrimeFunc{\pattern} \in \{A_1,\dots,A_{n-m+1}\}\), 
	which can be performed in \(\BigONotation(\log n)\) time using a balanced BST.

\end{proof}

The general statement above can be specialized to different regimes of sparsity.  
First, if the number of non-wildcard positions is bounded by a sublinear power of \(n\), we obtain the following corollary.

\begin{corollary}
	\label{cor:sparse}
	If the number of non-wildcard positions satisfies \(\nonWild \leq n^{1-\delta}\) for some constant \(0 < \delta < 1\), 
	then Algorithm~\ref{alg:sparse-pattern} runs in preprocessing time \(\BigONotation(n^{2-\delta})\) and supports updates in \(\BigONotation(n^{1-\delta})\) time.  
\end{corollary}

Moreover, since the number of non-wildcard positions can never exceed the pattern length \(m\), we can directly derive the following bound.

\begin{corollary}
	\label{cor:dense}
	Since \(\nonWild \leq m\), the algorithm also admits the trivial upper bound obtained by setting \(\nonWild = m\).  
	In this case, Algorithm~\ref{alg:sparse-pattern} requires preprocessing time \(\BigONotation(n \cdot m)\) and update time \(\BigONotation(m + \log n)\).  
\end{corollary}

	\section{Hardness of \OurProb}

\label{sec:hardness}

In this section, we establish a conditional lower bound for \OurProb\ through a reduction from the \OV problem. 
Our goal is to show that any significant improvement in the query time for \OurProb\ would imply a breakthrough for \OV, 
which is widely believed to be unlikely under the Strong Exponential Time Hypothesis ($\mathsf{SETH}$).

\begin{problem}[The Orthogonal Vectors Problem]
	Given a set \( A \) of \( \lengthT \) vectors from \( \{0, 1\}^d \), the \OV\ problem asks whether there exists a pair of distinct vectors 
	\( U, V \in A \) such that their inner product is zero:
	$$
	\sum_{h=1}^d U_h \cdot V_h = 0.
	$$
	That is, no coordinate position contains a \(1\) in both \(U\) and \(V\).
\end{problem}
Williams \cite{williams2005new} establishes the following connection between \OV\ and $\mathsf{SETH}$.
\begin{theorem}[\cite{williams2005new}]
	\label{thm:ov}
	If there exists \( \varepsilon > 0 \) such that, for all constants \( c \), the \OV\ problem on a set of \( \lengthT \) vectors of 
	dimension \( d = c \log \lengthT \) can be solved in 
	\( 2^{\LittleONotation(d)} \cdot \lengthT^{2-\varepsilon} \) time, 
	then $\mathsf{SETH}$ is false.
\end{theorem}

We now show a reduction of  \OV\  to  \OurProb. 
This reduction forms the backbone of our lower bound.


\begin{lemma}
	\label{lemma:ov}
	An instance of \OV\ with a set \( A \) of \( \lengthT \) vectors in dimension \( d \) can be reduced to an instance of \OurProb\ with 
	pattern length \( \lengthP = d+2 \) and text length \( \BigONotation(\lengthT \cdot d) \). 
	If \OurProb\ can be solved with preprocess time \( h(\lengthT) \) and query time \( g(\lengthT) \), then \OV\ can be solved in 
	\( \BigONotation(h(d \cdot \lengthT) + d \cdot \lengthT \cdot g(d \cdot \lengthT)) \) time.
\end{lemma}

\begin{proof}
	Let \( A = \{V_1, V_2, \ldots, V_{\lengthT}\} \), where each \( V_i \) is a \( d \)-dimensional binary vector.  
	For each vector \( V_i \), construct a modified vector \( V'_i \) by replacing every \( 1 \) in \( V_i \) with \( 0 \), and every \( 0 \) in \( V_i \) with a wildcard symbol.	
	Define the pattern and text as
	$$
	\pattern = \# V'_1 \# , 
	\quad 
	\textstr = \# V_1 \# V_2 \# \cdots \# V_{\lengthT} \#,
	$$
	where \(\#\) is a delimiter not in \(\{0,1,\wildcard\}\).  
	By construction, \(\pattern\) has length \( \lengthP = d+2 \) and \(\textstr\) has length \( \BigONotation(\lengthT \cdot d) \).
	We claim that \(\pattern\) matches a substring of \(\textstr\) if and only if the \OV\ instance contains an orthogonal pair.
	Indeed, suppose \(\pattern = \# V'_i \#\) matches \(\# V_j \#\) in \(\textstr\).  
	Then, for each coordinate \(k\):
	\begin{itemize}
		\item If \( V_{i,k} = 1 \), then \( V'_{i,k} = 0 \). Since \(0\) can only match \(0\), we must have \(V_{j,k} = 0\).
		\item If \( V_{i,k} = 0 \), then \( V'_{i,k} = \wildcard \), which matches both \(0\) and \(1\).
	\end{itemize}
	Thus, there is no coordinate where both \( V_{i,k} = 1 \) and \( V_{j,k} = 1 \), which means \( V_i \) and \( V_j \) are orthogonal.  
	The converse follows by reversing this reasoning.
	
	To check all pairs, we iteratively transform the pattern into each \(\# V'_i \#\) for \( i = 1, \ldots, \lengthT \).  
	Each transformation requires at most \(d\) symbol changes, followed by a query to check if the pattern matches a substring of \(\textstr\).  
	Over all \( \lengthT \) vectors, this yields \( \BigONotation(d \cdot \lengthT) \) updates and queries.  
	Each query costs \( g(d \cdot \lengthT) \), so the total runtime is
	$
	\BigONotation(d \cdot \lengthT \cdot g(d \cdot \lengthT)).
	$
\end{proof}

The reduction allows us to transfer lower bounds from \OV\ to \OurProb.

\begin{theorem}
	\label{thm:lower-bound}
	Assuming $\mathsf{SETH}$ holds, the fully dynamic wildcard pattern matching problem with 
	\(k = \Omega(\log \lengthT)\) wildcards cannot be solved with preprocessing time 
	\(\BigONotation(\lengthT^{2-\delta})\) and query time \(\BigONotation(\lengthT^{1-\varepsilon})\) 
	for any constants \(\delta, \varepsilon > 0\).
\end{theorem}

\begin{proof}
	Suppose, for the sake of contradiction, that \OurProb\ admits an algorithm with 
	query time \( g(\lengthT) = \BigONotation(\lengthT^{1-\varepsilon}) \) and preprocessing time 
	\( h(\lengthT) = \BigONotation(\lengthT^{2-\delta}) \), for some constants \( \varepsilon,\delta > 0 \).  
	
	By Lemma~\ref{lemma:ov}, an instance of the \OV\ problem with \(\lengthT\) vectors in dimension \(d\) can be reduced to an instance of \OurProb\ of text length \(\BigONotation(d \cdot \lengthT)\) and pattern length \(d+2\).  
	The reduction implies that \OV\ can then be solved in
	$
	\BigONotation\!\left(h(d \cdot \lengthT) + d \cdot \lengthT \cdot g(d \cdot \lengthT)\right).
	$
	Substituting the assumed bounds for \(h(\cdot)\) and \(g(\cdot)\), we obtain
	$$
	\BigONotation\!\left((d \cdot \lengthT)^{2-\delta} + d \cdot \lengthT \cdot (d \cdot \lengthT)^{1-\varepsilon}\right)
	= \BigONotation\!\left(d^{2-\delta}\cdot \lengthT^{2-\delta} + d^{2-\varepsilon}\cdot \lengthT^{2-\varepsilon}\right).
	$$
	
	Now set \(d = c \log \lengthT\) for a fixed constant \(c > 0\).  
	Since \(d = \BigONotation(\log \lengthT)\), both \(d^{2-\delta}\) and \(d^{2-\varepsilon}\) grow subexponentially in \(d\), i.e.,
	$$
	d^{2-\delta} = 2^{\LittleONotation(d)} 
	\quad\text{and}\quad 
	d^{2-\varepsilon} = 2^{\LittleONotation(d)}.
	$$
	Thus the total running time becomes
	$$
	\BigONotation\!\left(2^{\LittleONotation(d)} \cdot \lengthT^{2-\delta} + 2^{\LittleONotation(d)} \cdot \lengthT^{2-\varepsilon}\right)
	= \BigONotation\!\left(2^{\LittleONotation(d)} \cdot \lengthT^{2-\min\{\delta,\varepsilon\}}\right).
	$$
	
	This runtime contradicts Theorem~\ref{thm:ov}, which asserts that such an improvement for \OV\ would refute $\mathsf{SETH}$.  
	Therefore, under $\mathsf{SETH}$, no such algorithm for \OurProb\ can exist.
\end{proof}

	\bibliographystyle{apalike}
	\bibliography{draft}

\begin{thebibliography}{}

\bibitem[Abrahamson, 1987]{abrahamson1987generalized}
Abrahamson, K. (1987).
\newblock Generalized string matching.
\newblock {\em SIAM journal on Computing}, 16(6):1039--1051.

\bibitem[Abrishami et~al., 2013]{Abrishami2013Bioinformatics}
Abrishami, V., Zald{\'\i}var-Peraza, A., de~la Rosa-Trev{\'\i}n, J.~M., Vargas,
  J., Ot{\'o}n, J., Marabini, R., Shkolnisky, Y., Carazo, J.~M., and Sorzano,
  C. O.~S. (2013).
\newblock A pattern matching approach to the automatic selection of particles
  from low-contrast electron micrographs.
\newblock {\em Bioinformatics}, 29(19):2460--2468.

\bibitem[Aho and Corasick, 1975]{aho1975efficient}
Aho, A.~V. and Corasick, M.~J. (1975).
\newblock Efficient string matching: an aid to bibliographic search.
\newblock {\em Communications of the ACM}, 18(6):333--340.

\bibitem[Alstrup et~al., 2000]{alstrup2000pattern}
Alstrup, S., Brodal, G.~S., and Rauhe, T. (2000).
\newblock Pattern matching in dynamic texts.
\newblock In {\em Proceedings of the eleventh annual ACM-SIAM symposium on
  Discrete algorithms}, pages 819--828.

\bibitem[Amir et~al., 2007a]{amir2007dynamic}
Amir, A., Landau, G.~M., Lewenstein, M., and Sokol, D. (2007a).
\newblock Dynamic text and static pattern matching.
\newblock {\em ACM Transactions on Algorithms (TALG)}, 3(2):19--es.

\bibitem[Amir et~al., 2007b]{Amir2007}
Amir, A., Landau, G.~M., Lewenstein, M., and Sokol, D. (2007b).
\newblock Dynamic text and static pattern matching.
\newblock {\em ACM Trans. Algorithms}, 3(2):19.

\bibitem[Amir et~al., 2004]{amir2004faster}
Amir, A., Lewenstein, M., and Porat, E. (2004).
\newblock Faster algorithms for string matching with k mismatches.
\newblock {\em Journal of Algorithms}, 50(2):257--275.

\bibitem[Amit et~al., 2014]{IEEE6698319}
Amit, M., Backofen, R., Heyne, S., Landau, G.~M., Möhl, M., Otto, C., and
  Will, S. (2014).
\newblock Local exact pattern matching for non-fixed rna structures.
\newblock {\em IEEE/ACM Transactions on Computational Biology and
  Bioinformatics}, 11(1):219--230.

\bibitem[Aronov et~al., 2024]{aronov2024general}
Aronov, B., Cardinal, J., Dallant, J., and Iacono, J. (2024).
\newblock A general technique for searching in implicit sets via function
  inversion.
\newblock In {\em 2024 Symposium on Simplicity in Algorithms (SOSA)}, pages
  215--223. SIAM.

\bibitem[Backurs and Indyk, 2014]{Backurs2015}
Backurs, A. and Indyk, P. (2014).
\newblock Edit distance cannot be computed in strongly subquadratic time
  (unless {SETH} is false).
\newblock {\em arXiv preprint arXiv:1412.0348}.

\bibitem[Barton, 2014]{Barton2014}
Barton, C. (2014).
\newblock On the average-case complexity of pattern matching with wildcards.
\newblock {\em arXiv preprint arXiv:1407.0950}.

\bibitem[Bathie et~al., 2024a]{Bathie2024}
Bathie, G., Charalampopoulos, P., and Starikovskaya, T. (2024a).
\newblock Approximate pattern matching with mismatches and wildcards.
\newblock {\em arXiv preprint arXiv:2402.07732}.

\bibitem[Bathie et~al., 2024b]{bathie2024pattern}
Bathie, G., Charalampopoulos, P., and Starikovskaya, T. (2024b).
\newblock Pattern matching with mismatches and wildcards.
\newblock {\em arXiv preprint arXiv:2402.07732}.

\bibitem[Bille et~al., 2022]{bille2022gapped}
Bille, P., G{\o}rtz, I.~L., Lewenstein, M., Pissis, S.~P., Rotenberg, E., and
  Steiner, T.~A. (2022).
\newblock Gapped string indexing in subquadratic space and sublinear query
  time.
\newblock {\em arXiv preprint arXiv:2211.16860}.

\bibitem[Boyer and Moore, 1977]{boyer1977}
Boyer, R.~S. and Moore, J.~S. (1977).
\newblock A fast string searching algorithm.
\newblock {\em Communications of the ACM}, 20(10):762--772.

\bibitem[Bringmann and K{\"u}nnemann, 2015]{Bringmann2015}
Bringmann, K. and K{\"u}nnemann, M. (2015).
\newblock Quadratic conditional lower bounds for string problems and dynamic
  time warping.
\newblock In {\em 56th Annual IEEE Symposium on Foundations of Computer Science
  (FOCS)}, pages 79--97. IEEE Computer Society.

\bibitem[Bushong et~al., 2020]{Bushong2020RACS}
Bushong, V., Sanders, R., Curtis, J., Du, M., Cerny, T., Frajtak, K., Bures,
  M., Tisnovsky, P., and Shin, D. (2020).
\newblock On matching log analysis to source code: A systematic mapping study.
\newblock In {\em Proceedings of the International Conference on Research in
  Adaptive and Convergent Systems (RACS)}, pages 181--187.

\bibitem[Capacho et~al., 2017]{Capacho2017}
Capacho, J.~V., Subias, A., Trave-Massuy{\`e}s, L., and Jimenez, F. (2017).
\newblock Alarm management via temporal pattern learning.
\newblock {\em Engineering Applications of Artificial Intelligence},
  65:506--516.

\bibitem[Chan et~al., 2023]{chan2023faster}
Chan, T.~M., Jin, C., Williams, V.~V., and Xu, Y. (2023).
\newblock Faster algorithms for text-to-pattern hamming distances.
\newblock In {\em 2023 IEEE 64th Annual Symposium on Foundations of Computer
  Science (FOCS)}, pages 2188--2203. IEEE.

\bibitem[Charalampopoulos et~al., 2020]{charalampopoulos2020dynamic}
Charalampopoulos, P., Gawrychowski, P., and Pokorski, K. (2020).
\newblock {Dynamic Longest Common Substring in Polylogarithmic Time}.
\newblock In Czumaj, A., Dawar, A., and Merelli, E., editors, {\em 47th
  International Colloquium on Automata, Languages, and Programming (ICALP
  2020)}, volume 168 of {\em Leibniz International Proceedings in Informatics
  (LIPIcs)}, pages 27:1--27:19, Dagstuhl, Germany. Schloss Dagstuhl --
  Leibniz-Zentrum für Informatik.

\bibitem[Cheng et~al., 2013]{Cheng2013}
Cheng, Y., Izadi, I., and Chen, T. (2013).
\newblock Pattern matching of alarm flood sequences by a modified
  smith--waterman algorithm.
\newblock {\em Chemical Engineering Research and Design}, 91(6):1085--1094.

\bibitem[Clifford and Clifford, 2007]{Clifford2007}
Clifford, P. and Clifford, R. (2007).
\newblock Simple deterministic wildcard matching.
\newblock {\em Inf. Process. Lett.}, 101(2):53--54.

\bibitem[Clifford et~al., 2018]{Clifford2018}
Clifford, R., Grønlund, A., Larsen, K.~G., and Starikovskaya, T. (2018).
\newblock Upper and lower bounds for dynamic data structures on strings.
\newblock {\em arXiv preprint arXiv:1802.06545}.

\bibitem[Cole and Hariharan, 2002a]{cole2002}
Cole, R. and Hariharan, R. (2002a).
\newblock Approximate string matching:a simpler faster algorithm.
\newblock {\em SIAM Journal on Computing}, 31(6):1761--1782.

\bibitem[Cole and Hariharan, 2002b]{ColeHariharan2002}
Cole, R. and Hariharan, R. (2002b).
\newblock Verifying candidate matches in sparse and wildcard matching.
\newblock In {\em Proceedings of the 34th Annual ACM Symposium on Theory of
  Computing (STOC '02)}, pages 592--601. ACM.

\bibitem[Cormen et~al., 2022]{CLRS}
Cormen, T.~H., Leiserson, C.~E., Rivest, R.~L., and Stein, C. (2022).
\newblock {\em Introduction to Algorithms}.
\newblock MIT Press, 4th edition.

\bibitem[Crochemore et~al., 2015]{Crochemore2015}
Crochemore, M., Iliopoulos, C.~S., Pissis, S.~P., and Tischler, G. (2015).
\newblock Online pattern matching with wildcards and dynamic updates.
\newblock pages 165--176. Springer.

\bibitem[Farach, 1997]{farach1997optimal}
Farach, M. (1997).
\newblock Optimal suffix tree construction with large alphabets.
\newblock In {\em Proceedings 38th Annual Symposium on Foundations of Computer
  Science}, pages 137--143. IEEE.

\bibitem[Faro and Lecroq, 2013]{Faro2013CSUR}
Faro, S. and Lecroq, T. (2013).
\newblock The exact online string matching problem: A review of the most recent
  results.
\newblock {\em ACM Computing Surveys}, 45(2).

\bibitem[Fenwick, 1994]{fenwick1994}
Fenwick, P.~M. (1994).
\newblock A new data structure for cumulative frequency tables.
\newblock {\em Software: Practice and Experience}, 24(3):327--336.

\bibitem[Ferragina and Grossi, 1995]{ferragina1995fast}
Ferragina, P. and Grossi, R. (1995).
\newblock Fast incremental text editing.
\newblock In {\em Proceedings of the sixth annual ACM-SIAM symposium on
  Discrete algorithms}, pages 531--540.

\bibitem[Fischer and Paterson, 1974]{Fischer1974}
Fischer, M.~J. and Paterson, M.~S. (1974).
\newblock String matching and other products.
\newblock In Karp, R.~M., editor, {\em Complexity of Computation}, volume~7 of
  {\em SIAM-AMS Proceedings}, pages 113--125, Providence, RI. AMS.

\bibitem[Galil and Park, 1990]{Galil1990}
Galil, Z. and Park, K. (1990).
\newblock An improved algorithm for approximate string matching.
\newblock {\em SIAM Journal on Computing}, 19(6):989--999.

\bibitem[Gog and Navarro, 2017]{Gog2017}
Gog, S. and Navarro, G. (2017).
\newblock Efficient dynamic pattern matching with wildcards in compressed
  texts.
\newblock {\em ACM Journal of Experimental Algorithmics}, 22:1.1--1.25.

\bibitem[Gu et~al., 1994]{GuFarachBeigel1994}
Gu, M., Farach, M., and Beigel, R. (1994).
\newblock An efficient algorithm for dynamic text indexing.
\newblock In {\em Proceedings of the Fifth Annual ACM-SIAM Symposium on
  Discrete Algorithms (SODA)}, pages 697--704.

\bibitem[Hakak et~al., 2019]{IEEE8703383}
Hakak, S.~I., Kamsin, A., Shivakumara, P., Gilkar, G.~A., Khan, W.~Z., and
  Imran, M. (2019).
\newblock Exact string matching algorithms: Survey, issues, and future research
  directions.
\newblock {\em IEEE Access}, 7:69614--69637.

\bibitem[Indyk, 1998]{Indyk1998}
Indyk, P. (1998).
\newblock Faster algorithms for string matching problems: Matching the
  convolution bound.
\newblock In {\em 39th Annual Symposium on Foundations of Computer Science
  (FOCS)}, pages 166--173, Washington, DC, USA. IEEE Computer Society.

\bibitem[Jin and Xu, 2024]{jin2024shaving}
Jin, C. and Xu, Y. (2024).
\newblock Shaving logs via large sieve inequality: Faster algorithms for sparse
  convolution and more.
\newblock In {\em Proceedings of the 56th Annual ACM Symposium on Theory of
  Computing}, pages 1573--1584.

\bibitem[Johannesmeyer et~al., 2002]{Johannesmeyer2002AIChE}
Johannesmeyer, M.~C., Singhal, A., and Seborg, D.~E. (2002).
\newblock Pattern matching in historical data.
\newblock {\em AIChE Journal}, 48(9):2022--2038.

\bibitem[Kalai, 2002]{Kalai2002}
Kalai, A. (2002).
\newblock Efficient pattern-matching with don't cares.
\newblock In {\em 13th Annual ACM-SIAM Symposium on Discrete Algorithms
  (SODA)}, pages 655--656, Philadelphia, PA, USA. Society for Industrial and
  Applied Mathematics.

\bibitem[K{\"a}rkk{\"a}inen and Sanders, 2003]{karkkainen2003simple}
K{\"a}rkk{\"a}inen, J. and Sanders, P. (2003).
\newblock Simple linear work suffix array construction.
\newblock In {\em International colloquium on automata, languages, and
  programming}, pages 943--955. Springer.

\bibitem[Karp and Rabin, 1987]{karp1987efficient}
Karp, R.~M. and Rabin, M.~O. (1987).
\newblock Efficient randomized pattern-matching algorithms.
\newblock {\em IBM Journal of Research and Development}, 31(2):249--260.

\bibitem[Knuth, 2013]{knuth2013art}
Knuth, D.~E. (2013).
\newblock {\em Art of Computer Programming, Volume 4, Fascicle 4, The:
  Generating All Trees--History of Combinatorial Generation}.
\newblock Addison-Wesley Professional.

\bibitem[Knuth et~al., 1977]{knuth1977}
Knuth, D.~E., Morris, J.~H., and Pratt, V.~R. (1977).
\newblock Fast pattern matching in strings.
\newblock {\em SIAM Journal on Computing}, 6(2):323--350.

\bibitem[Kociumaka et~al., 2019]{Kociumaka2019}
Kociumaka, T., Radoszewski, J., and Starikovskaya, T. (2019).
\newblock Longest common substring with approximately k mismatches.
\newblock {\em Algorithmica}, 81(6):2633--2652.

\bibitem[Kopelowitz and Porat, 2022]{Kopelowitz2022}
Kopelowitz, T. and Porat, E. (2022).
\newblock On the average-case complexity of pattern matching with wildcards.
\newblock {\em Theoretical Computer Science}, 930:92--108.

\bibitem[Landau and Vishkin, 1988]{Landau1988}
Landau, G.~M. and Vishkin, U. (1988).
\newblock Fast parallel and serial approximate string matching.
\newblock {\em Journal of Algorithms}, 10(2):157--169.

\bibitem[Levenshtein, 1966]{Levenshtein1966}
Levenshtein, V.~I. (1966).
\newblock Binary codes capable of correcting deletions, insertions and
  reversals.
\newblock {\em Soviet Physics Doklady}, 10(8):707--710.

\bibitem[Masek and Paterson, 1980]{Masek1980}
Masek, W.~J. and Paterson, M.~S. (1980).
\newblock A faster algorithm computing string edit distances.
\newblock {\em J. Comput. Syst. Sci.}, 20(1):18--31.

\bibitem[McCreight, 1976]{mccreight1976space}
McCreight, E.~M. (1976).
\newblock A space-economical suffix tree construction algorithm.
\newblock {\em Journal of the ACM (JACM)}, 23(2):262--272.

\bibitem[Monteiro and dos Santos, 2024]{Monteiro2024}
Monteiro, B. and dos Santos, V. (2024).
\newblock String matching with a dynamic pattern.
\newblock {\em arXiv preprint arXiv:2406.11318}.

\bibitem[Navarro, 2001]{Navarro2001}
Navarro, G. (2001).
\newblock A guided tour to approximate string matching.
\newblock {\em ACM Computing Surveys}, 33(1):31--88.

\bibitem[Neamatollahi et~al., 2020]{IEEE8967097}
Neamatollahi, P., Hadi, M., and Naghibzadeh, M. (2020).
\newblock Simple and efficient pattern matching algorithms for biological
  sequences.
\newblock {\em IEEE Access}, 8:23838--23846.

\bibitem[Needleman and Wunsch, 1970]{Needleman1970}
Needleman, S.~B. and Wunsch, C.~D. (1970).
\newblock A general method applicable to the search for similarities in the
  amino acid sequences of two proteins.
\newblock {\em J. Mol. Biol.}, 48(3):444--453.

\bibitem[Pandiselvam et~al., 2014]{Pandiselvam2014Comparative}
Pandiselvam, P., Marimuthu, T., and Lawrance, R. (2014).
\newblock A comparative study on string matching algorithms of biological
  sequences.
\newblock {\em arXiv preprint arXiv:1401.7416}.

\bibitem[Sahinalp and Vishkin, 1996]{sahinalp1996efficient}
Sahinalp, S.~C. and Vishkin, U. (1996).
\newblock Efficient approximate and dynamic matching of patterns using a
  labeling paradigm.
\newblock In {\em Proceedings of 37th Conference on Foundations of Computer
  Science}, pages 320--328. IEEE.

\bibitem[Schwartz, 1980]{schwartz1980fast}
Schwartz, J.~T. (1980).
\newblock Fast probabilistic algorithms for verification of polynomial
  identities.
\newblock {\em Journal of the ACM (JACM)}, 27(4):701--717.

\bibitem[Seidel and Aragon, 1996]{seidel1996randomized}
Seidel, R. and Aragon, C.~R. (1996).
\newblock Randomized search trees.
\newblock {\em Algorithmica}, 16(4):464--497.

\bibitem[Tahir, 2017]{Tahir2017ESWA}
Tahir, M. (2017).
\newblock Efficient pattern matching algorithm for dna sequences.
\newblock {\em Expert Systems with Applications}, 82:96--112.

\bibitem[Ukkonen, 1985]{Ukkonen1985}
Ukkonen, E. (1985).
\newblock Algorithms for approximate string matching.
\newblock {\em Information and Control}, 64(1-3):100--118.

\bibitem[Ukkonen, 1995]{ukkonen1995line}
Ukkonen, E. (1995).
\newblock On-line construction of suffix trees.
\newblock {\em Algorithmica}, 14(3):249--260.

\bibitem[von~zur Gathen and Gerhard, 2013]{vzGG}
von~zur Gathen, J. and Gerhard, J. (2013).
\newblock {\em Modern Computer Algebra}.
\newblock Cambridge University Press, 3rd edition.

\bibitem[Weiner, 1973a]{weiner1973linear}
Weiner, P. (1973a).
\newblock Linear pattern matching algorithms.
\newblock In {\em 14th Annual Symposium on Switching and Automata Theory (swat
  1973)}, pages 1--11. IEEE.

\bibitem[Weiner, 1973b]{Weiner1973}
Weiner, P. (1973b).
\newblock Linear pattern matching algorithms.
\newblock In {\em 14th Annual Symposium on Switching and Automata Theory},
  pages 1--11, Washington, DC, USA. IEEE Computer Society.

\bibitem[Williams, 2005]{williams2005new}
Williams, R. (2005).
\newblock A new algorithm for optimal 2-constraint satisfaction and its
  implications.
\newblock {\em Theoretical Computer Science}, 348(2-3):357--365.

\bibitem[Wu et~al., 1990]{wu1990np}
Wu, S., Manber, U., Myers, G., and Miller, W. (1990).
\newblock An o (np) sequence comparison algorithm.
\newblock {\em Information Processing Letters}, 35(6):317--323.

\end{thebibliography}


\begin{thebibliography}{1}

\bibitem{DBLP:journals/cacm/Dijkstra68a}
Edsger~W. Dijkstra.
\newblock Letters to the editor: go to statement considered harmful.
\newblock {\em Commun. {ACM}}, 11(3):147--148, 1968.
\newblock \href {https://doi.org/10.1145/362929.362947}
  {\path{doi:10.1145/362929.362947}}.

\bibitem{DBLP:books/mk/GrayR93}
Jim Gray and Andreas Reuter.
\newblock {\em Transaction Processing: Concepts and Techniques}.
\newblock Morgan Kaufmann, 1993.

\bibitem{DBLP:conf/focs/HopcroftPV75}
{John E.} Hopcroft, {Wolfgang J.} Paul, and {Leslie G.} Valiant.
\newblock On time versus space and related problems.
\newblock In {\em 16th Annual Symposium on Foundations of Computer Science,
  Berkeley, California, USA, October 13-15, 1975}, pages 57--64. {IEEE}
  Computer Society, 1975.
\newblock \href {https://doi.org/10.1109/SFCS.1975.23}
  {\path{doi:10.1109/SFCS.1975.23}}.

\bibitem{DBLP:journals/cacm/Knuth74}
Donald~E. Knuth.
\newblock {Computer Programming as an Art}.
\newblock {\em Commun. {ACM}}, 17(12):667--673, 1974.
\newblock \href {https://doi.org/10.1145/361604.361612}
  {\path{doi:10.1145/361604.361612}}.

\end{thebibliography}
	
	\appendix
	
\end{document}